 \documentclass[12pt,  draftclsnofoot, onecolumn]{IEEEtran}
\usepackage[noadjust]{cite}
   
\usepackage{color}
\usepackage{graphicx}
\usepackage{amsthm}
\usepackage{algorithmic}

\usepackage{subcaption}

\usepackage{algorithm}

\usepackage[cmex10]{amsmath}
\usepackage{amssymb}
\usepackage{mathtools}
\usepackage{float}
\usepackage{tikz}
\usepackage{epstopdf}

\graphicspath{ {figures/} }

\newtheorem{theorem}{Theorem}
\newtheorem{lemma}{Lemma}
\newtheorem{corollary}{Corollary}
\newtheorem{definition}{Definition}
\newtheorem{remark}{Remark}

\newcommand{\saeed}[1]{}
\newcommand{\saeedLater}[1]{}

\newcommand{\saeedDissertation}[1]{}

\newcommand{\doc}[1]{}

\usetikzlibrary{positioning}
\usetikzlibrary{fit}

\tikzstyle{axis-line}=[thick, ->, >=stealth]
\tikzstyle{axis-tick}=[thin]
\tikzstyle{arrow} = [thick,->]

\tikzstyle{block}=[rectangle,black,draw, text centered, minimum height=0.7cm]
\tikzstyle{summul} = [draw, shape=circle, text centered,inner sep=0]

\title{The Method of Conditional Expectations \\for PAPR and Cubic Metric Reduction}
\date{Started 30.01.2018}
\author{Saeed Afrasiabi-Gorgani and Gerhard Wunder 
\thanks{This work was supported by the German Research Foundation (DFG) grant WU 598/3-1.}
\thanks{The authors are with the Department of Mathematics and Computer Science, Freie Universit\"{a}t Berlin, 14195, Berlin, Germany, e-mail: {s.afrasiabi, g.wunder}@fu-berlin.de.}}%

\begin{document}

\maketitle

\begin{abstract}
The OFDM waveform exhibits high fluctuation in the signal envelope which causes distortion in the nonlinear power amplifier of the transmitter. Peak-to-Average Power Ratio (PAPR) and Cubic Metric (CM) are the common metrics to quantify the phenomenon. A promising approach for PAPR or CM reduction is Sign Selection which is based on altering the signs of the data symbols. In this paper, the Method of Conditional Expectations (CE Method) is proposed to obtain a competing suboptimal solution to the Sign Selection problem. For PAPR reduction, a surrogate metric is introduced which allows for an efficient application of the CE Method. For CM reduction, the tractability of the definition of CM is exploited to this end. The algorithm is analyzed to obtain an upper bound on the worst-case reduced metric value. A noticeable characteristic  is the persistent reduction capability for a wide range of subcarrier numbers. In particular, simulations show a reduction of the so-called ``effective PAPR"  to about~6.5~dB from 10.5~dB and 11.7~dB respectively for~64 to~1024 subcarriers. A similar steady reduction of 3~dB is observed for CM.  In addition, the CE Method leads to a pruned version of  Sign Selection which halves the rate loss.
\end{abstract}


\section{Introduction}

\saeed{we cannot show that $\mathbb{E}[u(t,\mathbf{c}\odot \mathbf{x}^\ast)]=0$. But simulations show it's good enough, no increase in the signal mean value, i.e. in turn its power. This is not now mentioned at all. if a reviewer asks for it, then a proof might be unreachable. I started a note in my notebook }

\saeed{bring back dissertation and saeedLater comments once at least}

\saeed{it's necessary to show that the begiining point in seach is not important! otherwise I even have to mentine dhwat $\mathcal{M}'$ I have used!}

Orthogonal Frequency Division Multiplexing (OFDM) is a well-known multicarrier waveform which has been used in the major wireless communication systems. A main drawback of OFDM scheme is the high dynamic range of its signal envelope, which causes nonlinear distortion at the output of the power amplifier \cite{Pun2007}. In order to avoid the distortion, the so-called power back-off needs to be applied in the power amplifier. Consequently, the power amplifier operates with a low energy efficiency. Especially for mobile equipments where battery life is limited and power amplifiers cannot have a large linear range due to cost constraints, the problem is more pressing~\cite{Ekstrom2006}. It is therefore critical  to reduce the required power back-off.

 The problem is commonly formulated as the minimization of a metric which captures the physical phenomenon and determines the power back-off. The classical metric is the ratio of the peak instantaneous signal power to the average power  over consecutive signal segments  referred to as Peak-to-Average-Power-Ratio (PAPR) \cite{Pun2007}.   An alternative metric called Cubic Metric (CM), which is based on the energy in the nonlinear distortion, was more recently proposed and reported to predict the required back-off more accurately \cite{CMmotorola2004}. 



\saeedDissertation{ In evaluation of a distortion-less method for PAPR reduction, three factors are considered: rate loss, computational complexity and the reduction gain. When potentials of a novel method are being investigated, we look at the achievable reduction gain versus rate loss regardless of the complexity. Computational complexity becomes important when the method is considered for practical purposes. Then modifications are usually required for a lower complexity which might cause degradation in performance. That is, a trade-off in complexity versus reduction gain must be done for a specific rate loss. The new paradigm allows for choice of the intermediate metric as an extra degree of freedom in achieving better point in the reduction gain, complexity and rate loss space. In this paper all these three factors will be addressed. However, evaluation of the computational complexity depends on technologies and requirements of the day. Therefore, a sound treatment of this aspect requires a separate piece and type of work and we suffice to a crude evaluation. }

The PAPR reduction problem has been tackled by several approaches, which can be broadly categorized into two groups. Methods based on deliberately introduced distortion constitute one category, with Clipping and Filtering \cite{armstrong2002} as a well-known example. The second category consists of the distortionless methods which typically provide PAPR reduction at the expense of some reserved resources which incurs rate loss, such as Selected Mapping (SLM) \cite{543811}, Tone Reservation (TR) and Tone Injection (TI) \cite{Tellado:2000}. The methods differ significantly at least in terms of reduction gain, rate loss, transmission power and complexity. A comparison of the pros and cons requires a separate study as provided, for instance, in \cite{1421929}. A refreshed and fundamental review of the problem is as well provided in \cite{gerhard2013SPM}. 

The CM reduction problem, on the other hand, has received limited attention compared to PAPR.  In particular, very few of the already known methods from PAPR reduction research are examined for CM reduction, such as in~\cite{Behravan2011}, \cite{ZhuDescClip2013} and \cite{SiohanCM2006} for TR, Clipping and Filtering and SLM, respectively. It will be emphasized in this paper that CM has a more amenable mathematical structure, which indicates that there is room to improve on the performance and complexity of the back-off reduction problem by considering CM instead of PAPR, besides its reportedly higher accuracy.

Sign Selection is a promising distortionless approach based on altering the signs of the data symbols to reduce the PAPR, which has shown potentials for considerable reduction performance at the price of a rate loss equivalent to one bit per complex data symbol for each utilized sign variable \cite{Sharif2004constantPMEPR, Afrasiabi2015derandomized,TellamburaGuidedSS2018,Sharif2009sign, tellamburaCrossEntropy2008}.
Considering $N$ subcarriers, there are $2^N$ possible sign combinations, which implies an exponential complexity order for the optimal sign selection. This has motivated research for competing suboptimal solutions. Some proposals with noticeable performance include the application of the method of Conditional Probabilities in \cite{Sharif2004constantPMEPR, Afrasiabi2015derandomized}, a sign selection method guided by clipping noise in \cite{TellamburaGuidedSS2018}, a greedy algorithm in  \cite{Sharif2009sign} and a cross-entropy-based algorithm in \cite{tellamburaCrossEntropy2008}. In this work,  the method of Conditional Expectations (CE Method), originally proposed in fields of discrete mathematics and graph theory \cite{MitzenmacherUpfal2005}, is used to treat the Sign Selection problem to develop a simple algorithm with a competitive performance for both PAPR and CM reduction requiring only $\frac{N}{2}$ sign bits.

The core idea of the CE method is to treat the optimization variables, i.e. the signs of the complex data symbols, as random variables. This artificial randomness is then employed to optimize the signs using conditional expectations.  In addition to a direct application of the method to PAPR, a surrogate function referred to as  Sum-Exp (SE) is proposed to gain indirect PAPR reduction. Unlike the other metrics, SE has no physical interpretation and is not directly related to power back-off. However, it will be shown that its reduction results in the reduction of the PAPR with lower complexity. The CE method is also applied to CM reduction, where the benefit of the mathematical tractability of CM in deriving low complexity closed-form expressions is demonstrated. As a rather uncommon characteristic among the solutions of the Sign Selection problem in the literature,  an increasing reduction gain in PAPR and CM for increasing number of subcarriers is shown by simulations, which implies a roughly constant back-off for a large range of $N$.  Furthermore, the CE method allows the analysis of the reduction performance by providing upper-bounds on reduced PAPR and CM values for any combination of the data symbols.

\paragraph*{Notation} A random variable $X$ is distinguished from a realization $x$ by using upper and lower case letters, respectively. Vectors are shown by bold-face letters. For a vector $\mathbf{x}$, the notation $x_{m:n}$ is the compact form for $[x_m, x_{m+1},\ldots, x_n]$. The expected value of $Y$ with respect to the random variable $X$ is denoted by $\mathbb{E}_X[Y]$, where the subscript may be omitted if clear from the context. Cardinality of a set $\mathcal{S}$ is denoted by $|\mathcal{S}|$.

\section{Preliminaries}
\label{sec:pre}
In this section, the OFDM signal model as well as the definitions of the metrics PAPR, SE and CM are first presented. Then the Sign Selection problem is formalized and discussed.

\subsection{Signal Model}
Consider an OFDM scheme with $N$ subcarriers. Let $\mathcal{M}$ be the set of the complex-valued constellation points. The data symbols that modulate the subcarriers are equiprobably and independently generated with zero mean, which implies that $\sum_{x\in\mathcal{M}} x=0$. Accordingly, the random vector $\mathbf{B}\in \mathcal{M}^N$ denotes the vector of data symbols in an OFDM symbol.  Denoting the frequency separation of the first and the last subcarriers as $F_s$, the baseband continuous-time signal model for an OFDM symbol is
\begin{equation}
	u(t, \mathbf{B})=\frac{1}{\sigma_b\sqrt{N}}\sum_{k=0}^{N-1} B_k e^{i \frac{2\pi}{N} F_s k t} \quad\quad t\in[0,T),
	\label{eq:ofdmContSymbolDef}
\end{equation}
where $T=\frac{N}{F_s}$ and the signal power is normalized by $\sigma_b=\sqrt{\mathbb{E}[|u(t,\mathbf{B})|^2]}$. 
With the sampling frequency $L F_s$, where $L> 1$ is the oversampling factor, the discrete-time signal model for an OFDM symbol is 
\begin{align}
	s (n, \mathbf{B})\!&= u(\frac{n}{L F_s},\mathbf{B}) =\!\frac{1}{\sigma_b \sqrt{N}}\sum_{k=0}^{N-1} \!B_k e^{i \frac{2\pi}{LN} k n} \ \  n=0,1,\ldots, LN\!-\!1.
	\label{eq:ofdmSymbolDef}
\end{align}
The oversampling is necessary for reliable measurement of PAPR and CM from the discrete-time signal~\cite{WunderPeakValue2003, KimCubicMetric2016}.

\subsection{Peak to Average Power Ratio (PAPR)}
	
\begin{definition}
	The PAPR metric is a function of the  random data vector $\mathbf{B}\in\mathcal{M}^N$ and is defined as 
	\begin{equation}
	\theta_N(\mathbf{B}) = \max_{n=0,1,\ldots,LN-1} |s(n,\mathbf{B})|^2,
	\label{eqn:PAPRdef}
	\end{equation}
	where $s(n,\mathbf{B})$ is given in \eqref{eq:ofdmSymbolDef} and $L >1$ is the oversampling factor.
\end{definition}

It will be seen that the maximum operator in the definition of the PAPR makes the required derivations of the CE Method difficult. Here we propose the Sum-Exp (SE) metric, which will be  shown to be a suitable objective function to replace PAPR such that  a desirable indirect PAPR reduction is gained by SE reduction.

\begin{definition}
	The SE metric is a function of the random data vector $\mathbf{B}\in\mathcal{M}^N$ and is defined as
	\begin{align}
		\zeta_N(\mathbf{B})= \sum_{n=0}^{LN-1} e^{\kappa |s(n,\mathbf{B})|^2},
		\label{eqn:SEdef}
	\end{align}
	where $s(n,\mathbf{B})$ is given in \eqref{eq:ofdmSymbolDef}, $\kappa\geq 1$ is an adjustable parameter and  $L >1$ is the oversampling factor.
\end{definition}

 \saeedDissertation{best text found so far is the note ``Basic properties of the soft maximum" by John Cook} The SE metric is obtained from the \emph{log-sum-exp} function of the squared magnitude of the signal samples, i.e. $\log \sum_{n=0}^{LN-1} e^{|s(n,\mathbf{B})|^2}$, which  is a well-known approximation of the maximum function \cite{Boyd} since 
\begin{align*}
	\max_{i=0,\ldots,LN-1} |s(n,\mathbf{B})|^2 \leq \log \sum_{i=0}^{LN-1} e^{|s(n,\mathbf{B})|^2} \leq \max_{i=0,\ldots,LN-1} |s(n,\mathbf{B})|^2 + \log LN.
\end{align*}
The first inequality is strict unless $LN=1$ and approaches an equality as the maximum becomes larger relative to the rest of the samples,  while the second inequality holds when all values are equal. That is, the approximation improves when the spread of the amplitudes of the signal samples is larger. Therefore, high ratio of the peak power to the average power of the OFDM signal implies that log-sum-exp is likely to be an acceptable approximation for PAPR. Furthermore, it motivates the introduction of the scaling factor $\kappa \geq 1$ to modify the log-sum-exp function as $\frac{1}{\kappa}\log \sum_{i=0}^{LN-1} e^{\kappa|s(n,\mathbf{B})|^2}$ to increase the spread. The SE metric is obtained from the modified log-sum-exp function by omitting the monotonically increasing $\log$ function as well as the constant~$\kappa^{-1}$.

\saeedDissertation{  this LSE is no more an approximation of the PAPR. But it keeps a good correlation in a scatter plot. That might be the way to justify that the algorithm finally works. }

\subsection{Cubic Metric (CM)} 

CM \cite{CMmotorola2004} is based on the assumption of a third-order (cubic) polynomial model for the input-output relation of the power amplifier. That is, the output signal $v_o(t)$ for a passband input signal $v(t)$ is  assumed to be
\begin{equation*}
	v_o(t)=g_1 v(t) + g_3 v^3(t), \quad t\in\mathbb{R},
\end{equation*}
where the linear gain $g_1$ and the non-linear gain $g_3$ are constant and related to the amplifier design. While PAPR is based only on the peaks of the instantaneous power, CM directly captures the energy in the distortion term $v^3(t)$ and is calculated  as
\begin{equation*}
	\mathrm{CM}_\mathrm{dB}= \frac{\mathrm{RCM}_\mathrm{dB}[v(t)]-\mathrm{RCM}_\mathrm{dB}[v_\mathrm{ref}(t)]}{K_\mathrm{slp}}+K_\mathrm{bw},
\end{equation*}
where the subscript $\mathrm{dB}$ refers to the value in logarithmic scale and the Raw Cubic Metric (RCM) of a signal is defined as
\begin{equation}
	\mathrm{RCM}_\mathrm{dB}[v(t)]= 20 \log_{10}\left( \mathrm{rms}\!\left[\left(\frac{v(t)}{\mathrm{rms}[v(t)]}\right)^3 \right]\right).
	\label{eqn:RCMdef}
\end{equation}
The reference signal $v_\mathrm{ref}(t)$, the  slope factor $K_\mathrm{slp}$ and the bandwidth scaling factor $K_\mathrm{bw}$ \cite{CMmotorola} are independent of $v(t)$ and are not discussed here. The Root Mean Square (RMS) of a signal $v(t)$ over a large enough interval $U\subset \mathbb{R}$ is $\mathrm{rms}[v(t)]=\sqrt{\frac{1}{U} \int_U v^2(t) dt}$.

Consider that reduction of CM for $v(t)$ is essentially equivalent to reduction of its RCM. In addition, CM and RCM are constants calculated for the whole continuous-time passband signal, whereas practical reduction algorithms operate over individual discrete-time baseband OFDM symbols. Therefore, the discrete-time baseband version of the RCM of an OFDM symbol is actually used for CM reduction, as done in \cite{Behravan2011,ZhuDescClip2013,SiohanCM2006}, which is referred to  as Symbol RCM (SRCM) in this paper. 

\begin{definition}
	SRCM is a function of the random data vector $\mathbf{B}\in\mathcal{M}^N$ and is defined as
	\begin{align}
		\eta_{\scriptscriptstyle N}(\mathbf{B})=\frac{1}{LN} \sum_{n=0}^{LN-1} |s(n,\mathbf{B})|^6,
		\label{eqn:SRCM}
	\end{align}
	where $s(n,\mathbf{B})$ is given in \eqref{eq:ofdmSymbolDef} and $L>1$ is the oversampling factor.
\end{definition}

\saeedDissertation{a new note as a picture in the same folder about the passband to baseband conversion. saying that we should filter out the harmonics} 
In order to show the relation of RCM and SRCM, we shall first briefly discuss the baseband representation of $v^3(t)$. Let the baseband equivalent representation of $v(t)$ be $h(t)=\sum_{m=-\infty}^{\infty} u(t-mT,\mathbf{B}_m)$ as a function of complex data symbols $\mathbf{B}_m\in\mathcal{M}^N$ pertaining to consecutive OFDM symbols. By a suitable choice of the normalization factor, it follows from the standard procedure of passband to baseband conversion \cite{Benedetto1999} that $\mathrm{rms}[v(t)]=\mathrm{rms}[|h(t)|]= 1$.
\saeedDissertation{technically it's not an equality as rms was not defined as ensemble average} 
Ignoring the scaling factors, it can as well be shown that $h^\ast(t)|h(t)|^2$ is the baseband representation of the frequency component of $v^3(t)$ at the carrier frequency \cite{Benedetto1999}, where $h^{\ast}(t)$ is the complex conjugate of $h(t)$. Consequently, $\mathrm{RCM}[v(t)]=(\mathrm{rms}[v^3(t)])^2=A \left(\mathrm{rms}[|h(t)|^3]\right)^2$ for some scalar $A$ gives the RCM in terms of the baseband continuous signal. Next, the discrete-time version of $h(t)$ is $h(n)= \sum_{m=-\infty}^{\infty} s(n-mLN,\mathbf{B}_m)$. Replacing the summation with an integral in calculation of the RMS of a discrete-time signal, we have $\mathrm{rms}[|h(t)|^3]\simeq \mathrm{rms}[|h(n)|^3]$ given adequate oversampling.  Finally, RCM can be written as
\begin{align}
	\mathrm{RCM}[v(t)]&\simeq \lim_{K\to\infty} \frac{1}{K} \sum_{n=0}^{K} |h(n)|^6  \nonumber \\
	&=  \lim_{M\to\infty} \frac{1}{2MLN}\sum_{m=-M}^{M-1} \sum_{n=0}^{LN-1}|s(n-mLN,\mathbf{B}_m)|^6 \nonumber \\
	&= \lim_{M\to\infty} \frac{1}{2M}\sum_{m=-M}^{M-1} \eta_{\scriptscriptstyle N}(\mathbf{B}_m).
	\label{eqn:RCM-SRCM}
\end{align}
Therefore, RCM of the OFDM signal is the average of the SRCM values of the underlying OFDM symbols.

\saeedDissertation{
\begin{align}
	\mathrm{RCM}[v(t)]&\simeq  \mathbb{E} [\eta_{\scriptscriptstyle N}(\mathbf{B})],
\end{align}
The equation follows from the independence of the data vectors $\mathbf{B}_m$. }

\subsection{The Sign Selection Problem}

As introduced before, Sign Selection refers to altering the signs of the data symbols in an OFDM symbol in order to reduce a desired metric, such as PAPR and CM. Therefore,  for the constellation $\mathcal{M}$, $\log_2 |\mathcal{M}|-1$ bits per transmitted data symbol carry information and one bit is determined by the adopted Sign Selection algorithm. To perform the bit-to-symbol mapping in the transmitter, initially consider taking independently and equiprobably distributed random sign bits  to complete the $\log_2 |\mathcal{M}|$-bit blocks. This formulation helps analytical derivations in later sections and will be shortly shown not to affect the solution.  For the resulting vector of complex data symbols $\mathbf{b}\in\mathcal{M}^N$, the Sign Selection approach seeks a solution $\mathbf{x}^\ast$ for the problem 
\begin{equation}
	\min_{\mathbf{x}\in\{-1,1\}^N} f(\mathbf{b}\odot\mathbf{x}),
	\label{eqn:optimizationProblem}
\end{equation}
where $f(\cdot)\geq 0$ is a metric defined on the OFDM symbol and $\odot$ denotes element-wise multiplication of vectors.  Accordingly, $\mathbf{b}\odot\mathbf{x}^\ast$ will be the transmitted symbols. Considering that the solution space of~\eqref{eqn:optimizationProblem} grows exponentially with $N$, the objective of this paper is to  derive an efficient algorithm to obtain a suboptimal solution.

\begin{figure}[t]
	\centering
	\includegraphics[width=0.3\columnwidth]{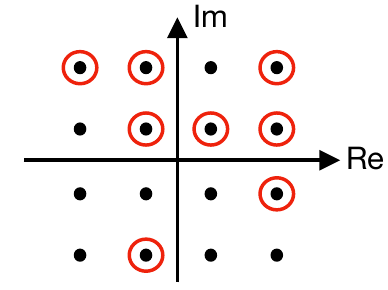}
	\caption{A non-unique choice of $\mathcal{C}$ from $\mathcal{M}$ for the 16-QAM constellation.}
	\label{fig:choiceOfConst}
\end{figure}

Now we justify that the random sign bits used to complete the $\log_2 |\mathcal{M}|$-bit blocks do not alter the minimization problem.  Assume that the constellation $\mathcal{M}$ is symmetric such that for each point $y\in\mathcal{M}$, the negated value $-y$ is in the set. Let $\mathcal{C}\subset\mathcal{M}$ be a non-unique choice of  $\frac{|\mathcal{M}|}{2}$ points of $\mathcal{M}$ such that if $y\in\mathcal{C}$,  then $-y\notin\mathcal{C}$.  A sample choice of $\mathcal{C}$ for 16-QAM is shown in Fig.~\ref{fig:choiceOfConst}.   For every $\mathbf{c} \in \mathcal{C}^N$, let $\Omega_\mathbf{c}=\{\mathbf{c} \odot \mathbf{x}, \mathbf{x}\in\{-1,1\}^N\}$.  The space of the data vectors $\mathcal{M}^N$ can be partitioned into the sets $\Omega_\mathbf{c}$ for $\mathbf{c}\in\mathcal{C}^N$ such that  
\begin{equation}
	\mathcal{M}^N=\cup_{\mathbf{c}\in\mathcal{C}^N}\ \Omega_\mathbf{c}
\end{equation}
and $\Omega_\mathbf{c} \cap \Omega_{\mathbf{c}'}=\varnothing$ when $\mathbf{c}\neq \mathbf{c}'$. Therefore, every $\mathbf{b}$ in \eqref{eqn:optimizationProblem} belongs to a partition $\Omega_\mathbf{d}$ such that $\mathbf{d}\in\mathcal{C}^N$ and $\mathbf{b}=\mathbf{d}\odot\mathbf{v}$ for some $\mathbf{v}\in\{-1,1\}^N$. Having all possible sign vectors as the solution space, it is clear that the Sign Selection problem always seeks the minimum of the partition which contains $\mathbf{b}$. Formally, $\min_{\mathbf{x}\in\{-1,1\}^N} f(\mathbf{b}\odot\mathbf{x})=\min_{\mathbf{x}\in\Omega_\mathbf{d}} f(\mathbf{x})$ for every $\mathbf{b}\in\Omega_\mathbf{d}$. Notice that although the starting vector $\mathbf{b}\in\Omega_\mathbf{d}$ does not affect the solution of~\eqref{eqn:optimizationProblem} for the partition $\Omega_\mathbf{d}$, it may change the suboptimal solution provided by a proposed algorithm.
 
The (bit-to-)symbol mapping in the transmitter and the decoding in the receiver are based on a predetermined $\mathcal{C}$.  On the transmitter side, the  data symbols are obtained by mapping $\log_2 |\mathcal{M}|-1$ bits to a point in $\mathcal{C}$. On the receiver side, the decoding of the symbol of each subcarrier is performed by choosing $c\in\mathcal{C}$ when one of $\pm c\in\mathcal{M}$ is detected and reversing the symbol mapping accordingly. Notice that the decoding adds no complexity to the receiver. Besides,  the choice of $\mathcal{C}$ plays a role only in the symbol mapping and decoding and is otherwise immaterial to the Sign Selection problem. Particularly, it can be shown that the partitioning described before is independent of $\mathcal{C}$.

As the final comment, sign selection clearly incurs rate loss. Consider the generalized scheme where $N_s\leq N$ signs are used in the sign selection. That is, $N_s$ data symbols carry $\log_2 |\mathcal{M}|-1$ bits of information and the remaining $N-N_s$ data symbols are mapped from $\log_2 |\mathcal{M}|$ bits to  $\mathcal{M}$. The incurred amount of rate loss, i.e. the ratio of the bits used for Sign Selection to the total number of bits in an OFDM symbol, is 
\begin{align}
	R= \frac{N_s }{N \log_2 |\mathcal{M}|}
	= \frac{N_s}{N} \log_{|\mathcal{M}|} 2.
	\label{eqn:rateLoss}
\end{align}
Evidently, the rate loss is inversely proportional to the constellation size $|\mathcal{M}|$.

\section{Method of Conditional Expectations}
\label{sec:algorithm}

The CE Method \saeed{spencer?}\cite{MitzenmacherUpfal2005} is represented here for obtaining a suboptimal solution to the Sign Selection problem for reduction of an arbitrary metric $f(.)\geq 0$. For a given data vector $\mathbf{b}\in\mathcal{M}^N$, a random vector of sign variables $\mathbf{X}\in \{-1,1\}^N$ is initially assumed with equiprobable and independent elements, which are then sequentially decided and fixed. Consider the $j^\mathrm{th}$ iteration where the random signs $X_{0:j-1}$ are fixed to $x^\ast_{0:j-1}$. The expected values of $f(\mathbf{b}\odot\mathbf{X})$ conditioned on $X_{0:j-1}=x^\ast_{0:j-1}$ with $X_j=1$ and $X_j=-1$ are compared and  the sign that yields the smaller expectation is chosen as $x_j^\ast$. Formally, a sub-optimal solution to the minimization problem stated in \eqref{eqn:optimizationProblem} can be obtained by sequentially choosing the sign variables as
	\begin{align}
		x_j^\ast &= 
		\begin{cases}
			\underset{x_0\in\{\pm1\}}{\mathrm{arg\ min}} \ \mathbb{E} [f(\mathbf{b}\odot\mathbf{X})|X_0=x_0] & j=0 \\
			\underset{x_j\in\{\pm1\}}{\mathrm{arg\ min}} \ \mathbb{E} [f(\mathbf{b}\odot\mathbf{X})|X_{0:j-1}=x^\ast_{0:j-1},X_j=x_j] & j=1,\ldots,N-1.
		\end{cases}
		\label{eqn:CEGeneralRule}
	\end{align}


The decision rule given in \eqref{eqn:CEGeneralRule} is based on introducing random sign variables and then reducing the conditional expectation of the original objective function. The justification that~\eqref{eqn:CEGeneralRule} leads to a desirable suboptimal solution of \eqref{eqn:optimizationProblem} is explained partly here for the general metric $f$ and will be concluded in Section~\ref{sec:analysis} for PAPR and SRCM. For the $j^\mathrm{th}$ sign decision, let
\begin{equation}
	g_j^\pm(\mathbf{b})=
	\begin{cases}
		\mathbb{E} [f(\mathbf{b}\odot\mathbf{X})|X_0=\pm1] & j=0 \\
		\mathbb{E} [f(\mathbf{b}\odot\mathbf{X})|X_{0:j-1}=x^\ast_{0:j-1},X_j=\pm1] & j=1\ldots,N-1.
	\end{cases}
	\label{eqn:gDefinition}
\end{equation}
Following the decision criterion in \eqref{eqn:CEGeneralRule}, we  have
\begin{align*}
	\mathbb{E} &[f(\mathbf{b}\odot\mathbf{X})|X_{0:j} =x^\ast_{0:j}]=\min\  \{g_j^+(\mathbf{b}), g_j^{-}(\mathbf{b})\},
\end{align*}
whereas for the $(j-1)$-th step with $j\geq1$, it holds that
\begin{align*}
	\mathbb{E} [f(\mathbf{b}\odot\mathbf{X}) | X_{0:j-1}=x^\ast_{0:j-1}] &=g_j^+(\mathbf{b}) \mathbb{P}(X_j=1|X_{0:j-1}=x^\ast_{0:j-1})   \nonumber \\
	& \quad \quad \quad \quad +g_j^-(\mathbf{b}) \mathbb{P}(X_j=-1|X_{0:j-1}=x^\ast_{0:j-1}) \nonumber \\
	&  =g_j^+(\mathbf{b}) \mathbb{P}(X_j=1)  +  g_j^-(\mathbf{b}) \mathbb{P}(X_j=-1) \nonumber \\
	& =\frac{1}{2}(g_j^+(\mathbf{b}) + g_j^-(\mathbf{b})) \nonumber \\
	&  \geq \min\{g_j^+(\mathbf{b}) , g_j^-(\mathbf{b})\}.
\end{align*} 
Therefore,
\begin{equation*}
	\mathbb{E}[f(\mathbf{b}\odot\!\mathbf{X})|X_{0:j}\!=\!x^\ast_{0:j}]\!\leq\! \mathbb{E}[f(\mathbf{b}\odot\!\mathbf{X})|X_{0:j-1}\!=\!x^\ast_{0:j-1}]
\end{equation*}
for $j=1,\ldots, N-1$. This shows that for a given $\mathbf{b}$, the non-increasing sequence of the conditional expectations begins with the \emph{initial expectation} $\mathbb{E}_\mathbf{X}[f(\mathbf{b}\odot\mathbf{X})]$ and ends with $f(\mathbf{b}\odot\mathbf{x}^\ast)=\mathbb{E} [f(\mathbf{b}\odot\mathbf{X})|\mathbf{X}=\mathbf{x}^\ast]$ where no randomness is left. That is, the last conditional expectation coincides with a metric value such that
\begin{align}
	f(\mathbf{b}\odot\mathbf{x}^\ast) \leq \mathbb{E}_\mathbf{X}[f(\mathbf{b}\odot\mathbf{X})].
	\label{eqn:CEguarantee}
\end{align}
\saeed{present this as a lemma?}
 This justifies that the decision criterion given in \eqref{eqn:CEGeneralRule} leads to a value of the original metric $f$ with the property stated above. Proving the reduction and the upper-bound on the reduced values is not known for the general case of the arbitrary metric $f$ and will be treated in Section~\ref{sec:analysis} specifically for PAPR and CM.  Calculation of the conditional expectations required at each step is a major step in development of the algorithm and  will be discussed in Section~\ref{sec:CE-calc}.

\saeedDissertation{say that we simply set $x^\ast_{0:N_f-1}$ as the sign of data symbols taken from $\mathcal{M}'$. this will hopefully keep a good consistency throughout the paper. then update the inequality guaranteed by the CE method to one which includes $N_f$.}

\saeedDissertation{do some work on the location of the signs, at least first or second half of subcarriers and their equivalence. }

\saeedDissertation{
\begin{lemma}
	The solution is independent of  the initial data vector in a partition. is this a property of the sign selection problem or the CE Method? I guess we can't claim that this is a property of any suboptimal solution of the SS problem. But it is a common property of the methods I know now. This is important as it determines where this lemma should be put.
\end{lemma}
from this lemma we need: 1) the decision rule, hence the CEs, updated using a given data vector which is from M and not M'. so we can do the proofs. NO we are okay with the updated SS problem. this is good to mention as a property. Finally I decided I don't use the lemma anywhere and I don't have time to write it :D the original problem in proof of Lemma1 is solved by the reformulation the SS problem}

\section{Calculation of the Conditional Expectations}
\label{sec:CE-calc}
For a given vector of data symbols $\mathbf{b}$, the decision for $x^\ast_j$ requires calculation of  $g^\pm_j(\mathbf{b})$ in~\eqref{eqn:gDefinition}  which is compactly rewritten as
\begin{align}
	g_j^\pm(\mathbf{b})=\mathbb{E} [f(\mathbf{b}\odot\mathbf{Y}^\pm_j)],
	\label{eqn:gDef2}
\end{align}
where 
\begin{equation*}
	\mathbf{Y}^\pm_j=[x^\ast_0,x^\ast_1,\ldots,x^\ast_{j-1},\pm1,X_{j+1},\ldots,X_{N-1}]^T
\end{equation*}
encapsulates the decided signs, the new sign variable set to $+1$ or $-1$ and the remaining random sign variables. The obvious way of calculating the conditional expectations for practically any metric $f$ is to use the empirical average $\hat{g}^\pm_j(\mathbf{b},\mathbf{X}^{1:Q})$ to estimate $g_j^\pm(\mathbf{b})$,   which is
\begin{equation}
	\hat{g}^\pm_j(\mathbf{b},\mathbf{X}^{1:Q})=\frac{1}{Q} \sum_{l=1}^{Q} f\left(\mathbf{b}\odot \psi_{j}^\pm(\mathbf{X}^l)\right),
	\label{eq:estimatorGeneral}
\end{equation}
where $Q$ is the number of realizations of the random sign vector used for the estimation and
\begin{equation}
	\psi_{j}^\pm(\mathbf{X}^l)=[x_0^\ast, \ldots, x_{j-1}^\ast, \pm 1, X^l_{0}, \ldots, X^l_{N-j-2}]^T,
	\label{eq:psiDef}
\end{equation}
where the random variables $X^l_k\in\{-1,1\}, l=1,2,\ldots,Q, k=0,1,\ldots,N-j-2$ are independent and equiprobable.

Deriving more efficient ways of calculating the conditional expectations $g^\pm_j(\mathbf{b})$ is a pivotal part of the proposed method. The PAPR metric does not lend itself well to mathematical manipulations to obtain closed-form expressions. Consequently, the conditional expectations are estimated by the sample average as in~\eqref{eq:estimatorGeneral}, which will be further discussed. On the contrary, the definitions of SRCM and SE together with the statistical properties of the signal samples $s(n,\mathbf{b}\odot\mathbf{Y}^\pm_j)$  make it possible to derive closed-form expressions for $g^\pm_j(\mathbf{b})$. These results depend on convergence of the signal samples in distribution to a Gaussian random variable, proof of which is not trivial due to the specific signal model imposed by the Sign Selection problem. This will be clarified in the second part of this section before treating the calculations for  SE and SRCM.

\subsection{PAPR metric}
\label{sec:PAPR}
As mentioned before, the available method for calculation of the conditional expectations of PAPR is to perform estimation as specified in~\eqref{eq:estimatorGeneral}. It is apt to study the estimator in terms of a relation between the amount of the required numerical computations, which is proportional to $Q$, and the performance. Although such analysis for PAPR was not reached, authors have presented interesting results in  \cite{afrasiabiWSA2016} for the closely related metric 
\begin{equation}
	\phi_N(\mathbf{b}) = \sqrt{\theta_N(\mathbf{b})},
\end{equation}
which is referred to as Crest Factor (CF). This is a valid alternative as firstly CF has the same physical meaning and practical significance and secondly its relation with PAPR  is monotonically increasing. In addition, simulations show almost identical PAPR reduction gained by reduction of CF. 

Accordingly, consider $g_j^\pm(\mathbf{b})$ in \eqref{eqn:gDef2} for $f(\mathbf{b}\odot\mathbf{Y}^\pm_j)=\phi_N(\mathbf{b}\odot\mathbf{Y}^\pm_j)$. The sample average with $Q$ realizations of the sign vector is
\begin{equation}
	\hat{g}^\pm_j(\mathbf{b},\mathbf{X}^{1:Q})=\frac{1}{Q} \sum_{l=1}^{Q} \phi_N\left(\mathbf{b}\odot \psi_{j}^\pm(\mathbf{X}^l)\right),
	\label{eq:estimator}
\end{equation}
where $\psi_{j}^\pm(\mathbf{X}^l)$ and the random vectors $\mathbf{X}^l, l=1,2,\ldots,Q$ were defined in \eqref{eq:psiDef}. 
\saeedDissertation{our results are for using same sign samples for both, otherwise it's considerably poorer for large N!} 
It is clear that $\mathbb{E}[\hat{g}^\pm_j(\mathbf{b},\mathbf{X}^{1:Q})]=g^\pm_j(\mathbf{b})$. Consequently, $\lim_{Q\to\infty} \hat{g}^\pm_j(\mathbf{b},\mathbf{X}^{1:Q})=g^\pm_j(\mathbf{b})$ as the variance of $\theta_N(\mathbf{b}\odot \psi_{j}^\pm(\mathbf{X}^l))$ is finite. \saeedDissertation{, although not known for finite $N$. saeed: but would having the variance do good characterisation?}  In order to obtain a relation between the reliability of the estimation and  $Q$, McDiarmid's concentration inequality \cite{mcdiarmid1998} was employed to bound the probability of deviation of the estimate  from its true value as stated in the following theorem \cite{afrasiabiWSA2016}.  \saeedDissertation{Concentration inequalities have as well been shown to be useful for characterizing ... in [].} The proof is provided in Appendix~\ref{app:mcdiarmidProof} for completeness.  \saeed{think about verifying the theorem by simulation in the last section}

\begin{theorem}
	\label{thm:mcdiarmid}
Consider the sample average $\hat{g}^{\pm}_j(\mathbf{b}, \mathbf{X}^{1:Q})$ given in~\eqref{eq:estimator} as an estimate of the conditional expectation $g^{\pm}_j(\mathbf{b})$ given in~\eqref{eqn:gDef2} with $f(\mathbf{b}\odot\mathbf{Y}^\pm_j)=\phi_N(\mathbf{b}\odot\mathbf{Y}^\pm_j)$. For any $\mathbf{b}\in\mathcal{M}^N$ and $\epsilon\geq 0$,
\begin{equation*}
	\mathbb{P} (|\hat{g}^{\pm}_j(\mathbf{b},\mathbf{X}^{1:Q})- g^{\pm}_j(\mathbf{b}) |\geq \epsilon) \leq 2 \exp\!\left(\!-2\epsilon^2\frac{ Q}{d^2} \frac{N}{(N-j-1)}\right),
	\label{eq:estimation-deviation}
\end{equation*}
where $d=2 \sigma_b^{-1}\max_{x\in\mathcal{M}} |x|$.
\end{theorem}

An interesting result of Theorem~\ref{thm:mcdiarmid} is that the upperbound on the probability of deviation is independent of $N$. This is further clarified as follows. A lower bound on the required $Q$ which guarantees the probability of deviation  by $\epsilon$ from the true value to be less than $p$ can be deduced within the context of Theorem~\ref{thm:mcdiarmid} as 
\begin{align*}
	Q_\circ &=  -\frac{d^2 \log \frac{p}{2}}{2\epsilon^2}  \frac{N-j-1}{N}.
\end{align*}
In particular, it indicates that $Q_\circ$ is proportional to the ratio of the number of the remaining sign variables to the total number of them. Equivalently,
\begin{align}
	Q_\circ &\approx  -\frac{d^2 \log \frac{p}{2}}{2\epsilon^2}  (1-\rho),
	\label{eqn:Qrel}
\end{align}
where $\rho=\frac{j}{N}$ and the approximation is due to $\frac{N-1}{N}\approx 1$ for large $N$.  However, establishing a connection between the probability of error in sign decision and $Q$ is challenging and needs further research. 

\subsection{Distribution of $s(n,\mathbf{b}\odot\mathbf{Y}^\pm_j)$}

We begin with characterizing the distribution of the continuous-time OFDM symbol $u(t, \mathbf{b}\odot \mathbf{Y}^\pm_j)$ in Theorem~\ref{thm:convergence}, which is required for performance analysis in Section~\ref{sec:analysis}. The distribution of the discrete-time version $s(n,\mathbf{b}\odot\mathbf{Y}^\pm_j)$ follows automatically, as stated in Corollary~\ref{cor:normalDist}, which is used in the derivation of the conditional expectations of SRCM and SE. 
As the first step, the following Lemma gives the covariance functions of the samples of the OFDM signal for a given~$\mathbf{b}$ and iteration $j$ of the CE Method as $N\to\infty$. Let
\begin{align}
	\hat{u}_r(t, \mathbf{b}\!\odot\! \mathbf{Y}^\pm_j) &\! =\! u_r(t, \mathbf{b}\!\odot\! \mathbf{Y}^\pm_j) - \mathbb{E}[u_r(t, \mathbf{b}\!\odot\! \mathbf{Y}^\pm_j)], \nonumber \\
	\hat{u}_i(t, \mathbf{b}\!\odot\! \mathbf{Y}^\pm_j) &\! =\! u_i(t, \mathbf{b}\!\odot\! \mathbf{Y}^\pm_j) - \mathbb{E}[u_i(t, \mathbf{b}\!\odot\! \mathbf{Y}^\pm_j)],
	\label{eqn:centeredSamples}
\end{align}
where subscripts $r$ and $i$ denote the real and imaginary parts respectively. 

\begin{lemma}
	\label{lem:variances}
	Consider $j=\rho N$ where $0\leq\rho\leq 1$ is a  rational number. For $\mathbf{B}$ randomly distributed in $\mathcal{M}^N$, let the variances and covariances of $u_r(t,\mathbf{B}\odot \mathbf{Y}^\pm_j)$ and $u_i(t, \mathbf{B}\odot \mathbf{Y}^\pm_j)$ with respect to $\mathbf{Y}^\pm_j$ as $N\to\infty$ and at any time instances $t_1, t_2\in [0,T)$  be denoted as
	\begin{align*}
		R^j_{rr}(\tau,\mathbf{B})&=\lim_{N\to\infty}\mathbb{E}_{\mathbf{Y}^\pm_j}[\hat{u}_r(t_1, \mathbf{B}\!\odot\! \mathbf{Y}^\pm_j) \hat{u}_r(t_2, \mathbf{B}\!\odot\! \mathbf{Y}^\pm_j)],  \nonumber \\
		R^j_{ri}(\tau,\mathbf{B})&=\lim_{N\to\infty}\mathbb{E}_{\mathbf{Y}^\pm_j}[\hat{u}_r(t_1, \mathbf{B}\!\odot\! \mathbf{Y}^\pm_j) \hat{u}_i(t_2, \mathbf{B}\!\odot\! \mathbf{Y}^\pm_j)],  \nonumber \\
		R^j_{ii}(\tau,\mathbf{B})&=\lim_{N\to\infty}\mathbb{E}_{\mathbf{Y}^\pm_j}[\hat{u}_i(t_1, \mathbf{B}\!\odot\! \mathbf{Y}^\pm_j) \hat{u}_i(t_2, \mathbf{B}\!\odot\! \mathbf{Y}^\pm_j)],
	\end{align*}
	where $\tau=t_2-t_1\in(-T,T)$. Then
	\begin{align*}
		&R^j_{rr}(\tau)= R^j_{ii}(\tau)= \frac{\sigma_b^2}{2}(\mathrm{sinc} (2F_s\tau)- \rho \mathrm{sinc} (2F_s\tau \rho)), \nonumber\\
		&R^j_{ri}(\tau)=
		\begin{dcases}
			\frac{\sigma_b^2}{2} \frac{1}{2\pi F_s \tau} (\frac{1}{\rho}\cos(2\pi F_s \rho \tau)\!-\! \cos(2\pi F_s \tau)) &\!\!\!\! \tau\neq 0 \\
			0 &\!\!\!\! \tau=0
		\end{dcases},
	\end{align*}
	with probability one. That is, the result holds for any $\mathbf{B}$ as $N\to\infty$ which is emphasized by omitting the argument $\mathbf{B}$ from the notation. Clearly, $R^j_{ri}(\tau)=R^j_{ir}(-\tau)$.
\end{lemma}

The proof is given in Appendix~\ref{app:variances}. The following theorem characterizes the distribution of the OFDM signal.

\begin{theorem}
	\label{thm:convergence}
	For $\mathbf{B}$ randomly distributed in $\mathcal{M}^N$ and $j=\rho N$ as specified in Lemma~\ref{lem:variances}, consider $\hat{u}(t,\mathbf{B}\odot \mathbf{Y}^\pm_j)$ as defined in \eqref{eqn:centeredSamples} at any set of time instances $\{t_1,t_2,\ldots,t_K\}\in [0,T)^K, K>1$. Omitting $\mathbf{B}\odot \mathbf{Y}^\pm_j$ to save space, the vector
	\begin{equation}
		[\hat{u}_r(t_1), \hat{u}_i(t_1), \hat{u}_r(t_2), \hat{u}_i(t_2), \ldots, \hat{u}_r(t_{K}), \hat{u}_i(t_{K})]^T
		\label{eqn:cramerWoldSeq1}
	\end{equation}
	converges  in distribution, as $N\to\infty$, to the vector
	\begin{equation}
		[x_1, y_1, x_2, y_2, \ldots, x_{K}, y_{K}]^T
		\label{eqn:cramerWoldSeq2}
	\end{equation}
	of jointly Gaussian random variables with $\mathbb{E}[x_m x_n]=R^j_{rr}(t_n-t_m)$, $\mathbb{E}[y_m y_n]=R^j_{ii}(t_n-t_m)$ and $\mathbb{E}[x_m y_n]=R^j_{ri}(t_n-t_m)$ as given in Lemma~\ref{lem:variances}.
\end{theorem}

\begin{proof}
The proof follows a standard procedure and is only outlined here. It essentially consists of the application of the Cramer-Wold theorem \cite{billingsley1999} to the vector in \eqref{eqn:cramerWoldSeq1} which requires that every linear combination of the elements of the vector in  \eqref{eqn:cramerWoldSeq1} converges in distribution to the same linear combination of the corresponding elements of the vector in \eqref{eqn:cramerWoldSeq2}. This can be verified by the Lindeberg condition. In this procedure, the existence of the covariances of the linear combination is  shown in Lemma~\ref{lem:variances}.
\end{proof}

From Theorem~\ref{thm:convergence}, the following result is immediate for the discrete-time OFDM signal at iteration $j$.
\begin{corollary}
	\label{cor:normalDist}
	For any given~$\mathbf{b}\in\mathcal{M}^N$, $n=0,1,\ldots,LN-1$ and $j=\rho N$ as defined in Lemma~\ref{lem:variances}, it holds that
	\begin{equation*}
		\begin{bmatrix}
			s_r(n, \mathbf{b}\!\odot\! \mathbf{Y}^\pm_j)- \mu_r(n,\mathbf{b}\!\odot\! \mathbf{Y}^\pm_j) \\
			s_i(n, \mathbf{b}\!\odot\! \mathbf{Y}^\pm_j)- \mu_i(n,\mathbf{b}\!\odot\! \mathbf{Y}^\pm_j)
		\end{bmatrix} 
		\xrightarrow{d} \mathcal{N}(0,\frac{1}{2}(1-\rho) I),
		\label{eqn:jointNormalS}
	\end{equation*}
	where $\xrightarrow{d}$ denotes convergence in distribution, $I$ is a $2\times2$ identity matrix and 
	\begin{align*}
		\mu_r(n,\mathbf{b}\!\odot\! \mathbf{Y}^\pm_j)&\!=\!\frac{1}{\sigma_b\sqrt{N}}\mathrm{Re}\left\{ \pm b_j e^{i\frac{2\pi}{LN}jn}\! +\sum_{k=0}^{j-1} b_k x_k^\ast e^{i\frac{2\pi}{LN}kn}  \right\}, \nonumber \\
		\mu_i(n,\mathbf{b}\!\odot\! \mathbf{Y}^\pm_j)&\!=\!\frac{1}{\sigma_b\sqrt{N}}\mathrm{Im}\left\{ \pm b_j e^{i\frac{2\pi}{LN}jn}\!+\sum_{k=0}^{j-1} b_k x_k^\ast e^{i\frac{2\pi}{LN}kn}  \right\}.
	\end{align*}
\end{corollary}

\begin{remark}
\label{rmk:sameInTheLimit}
	A pivotal result which enables the analytical derivations in the remainder of this paper is that at every iteration of the algorithm, the distribution of $\hat{u}(t,\mathbf{b}\odot\mathbf{Y}^\pm_j)$ in the limit is independent of $\mathbf{b}$. In addition, the distribution of $u(t,\mathbf{b}\odot\mathbf{X})$, i.e. prior to any sign decision, is identical  to that of $u(t,\mathbf{B})$ as $N\to\infty$.
\end{remark}

\begin{remark}
\label{remarkApprox}
In the following sections, the asymptotically Gaussian distribution shown in Corollary~\ref{cor:normalDist} is used to approximate  the distribution of $s(n,\mathbf{b}\odot\mathbf{Y}_j^\pm)$ for the finite but large enough number of random sign variables, i.e. $N-j-1$ at iteration $j$.  This can be used to derive closed-form expressions of the sign decision criterion \eqref{eqn:CEGeneralRule} only for $j=0,1,\ldots,N-N_e-1$. The number of the excluded signs $N_e$, for which the approximation is unacceptable, will be determined based on simulations in Section~\ref{sec:perf}. 
\end{remark}

\subsection{SE Metric}
\label{sec:LSEcalc}

By substituting $\zeta_N(\mathbf{b}\odot\mathbf{Y}^\pm_j)$  for $f(\mathbf{b}\odot\mathbf{Y}^\pm_j)$  in \eqref{eqn:gDef2}, we have
\begin{align}
	g^{\pm}_j(\mathbf{b})=  \sum_{n=0}^{LN-1} \mathbb{E} \left[e^{  \kappa |s(n,\mathbf{b}\odot\mathbf{Y}^\pm_j)|^2}\right].
	\label{eqn:CEsForSE}
\end{align}
It was shown in Corollary~\ref{cor:normalDist} that the real and imaginary components of $s(n,\mathbf{b}\odot \mathbf{Y}^\pm_j)$ are  Gaussian and independent in the limit with equal variances. For $j=\rho N$, let $\delta_j^2=R_{rr}(0)=R_{ii}(0)=\frac{1}{2}(1-\rho)$ as obtained in Lemma~\ref{lem:variances}. Here we apply the approximation suggested in Remark~\ref{remarkApprox} for $j=0, \ldots, N-N_e-1$. Specifically, the real and imaginary parts of
\begin{equation}
	z(n, \mathbf{b}\odot \mathbf{Y}^\pm_j)=\delta_j^{-1} s(n, \mathbf{b}\odot \mathbf{Y}^\pm_j), \ j=0,\ldots,N-N_e-1
	\label{eqn:Zdef}
\end{equation}
have approximately unit variances with accordingly scaled expected values. Therefore, $|z(n, \mathbf{b}\odot \mathbf{Y}^\pm_j)|^2$ for large enough $N-j$ is approximately a non-central $\chi^2$-distributed random variable with two degrees of freedom. Consider the moment generating function of $|z(n, \mathbf{b}\odot \mathbf{Y}^\pm_j)|^2$ which is 
\begin{align*}
	M^\pm_{j,n}(t)&= \mathbb{E}\left[e^{t |z(n, \mathbf{b}\odot \mathbf{Y}^\pm_j)|^2}\right] \\
	&=e^{\lambda^\pm_{j,n} t (1-2t)^{-1}} (1-2t)^{-1} \quad 2t<1,
\end{align*}
where  the non-centrality parameter $\lambda^\pm_{j,n}$ is
\begin{equation*}
	\lambda^\pm_{j,n}=\delta_j^{-2 }\left(\mu_r^2(n,\mathbf{b}\odot \mathbf{Y}^\pm_j)+\mu_i^2(n,\mathbf{b}\!\odot\! \mathbf{Y}^\pm_j)\right),
\end{equation*}
and $\mu_r$ and $\mu_i$ were given in Corollary~\ref{cor:normalDist}. It can be seen that the terms in~\eqref{eqn:CEsForSE} are identical to the definition of $M^\pm_{j,n}(\kappa\delta_j^2)$. Consequently,
\begin{align}
	g^{\pm}_j(\mathbf{b})=(\kappa\delta_j^2)^{-1}\beta \sum_{n=0}^{LN-1} e^{\beta \lambda^\pm_{j,n}},\quad j=0,\ldots,N-N_e-1,
	\label{eqn:LSEcalc}
\end{align}
where $\beta=\kappa\delta_j^2(1-2\kappa\delta_j^2)^{-1}$. \saeedDissertation{ It is worth mentioning that in derivations for both SRCM and LSE we have used a standard chi-squared variable, hence the same $\lambda_n$, but in SRCM it was possible to establish the connection between the original quantity to be calculated and the MGF by a scalar. Here since we are not working with expected value of the r.v. itself, but it's exp, we have to do the job by choosing $\gamma$.} Finally, a closed-form decision rule can be obtained as
\begin{align}
	x_j^\ast=-\mathrm{sign} \left[\sum_{n=0}^{LN-1} \left(e^{\beta \lambda^+_{j,n}} -  e^{\beta \lambda^-_{j,n}}\right)\right], \quad j=0,\ldots,N-N_e-1.
	\label{eqn:LSErule}
\end{align}
The number of the last sign decisions which do not follow  the closed-form expression in \eqref{eqn:LSEcalc}, i.e. $N_e$, will be determined in Section~\ref{sec:perf}. A sample average must be inevitably used instead for signs $j=N-N_e, \ldots, N-1$ as in~\eqref{eq:estimatorGeneral}.

\saeedLater{sth wrong here. LSE with the log is quite good, for very large or small $N$, it gives a very good approximation with a tiny MSE. but it's not the case at all without the log. why? and then why the algorithm works without the log?!!!}

\begin{algorithm}[t]
\caption{Implementation of the CE Method for CM reduction by Sign Selection. \saeed{make connections to between equations derived and the calculations here} \saeed{make it for SE as it comes first and is more general}}
	\label{alg:ce-cm}
	\begin{algorithmic}[1]
		\REQUIRE $b_0, \ldots, b_{N_f-1}$: data symbols with $\log_2 |\mathcal{M}|$-bit mapping  \\
		$b_{N_f}, \ldots, b_{N-1}$: data symbols with $(\log_2 |\mathcal{M}| - 1$)-bit mapping  \\
		\STATE $\mathbf{x}^\ast \gets [1, 1, \ldots, 1]_{N\times 1}$
		\STATE $\mathbf{n} \gets [0, 1, 2,\ldots, LN-1]$
		\STATE $\mathbf{h} \gets \sum_{j=0}^{N_f-1} b_j \exp(i 2\pi j \frac{\mathbf{n}}{LN})$ 
		\\\COMMENT{Element-wise operations on arrays are assumed in this pseudocode.}
		\FOR {$j = N_f$ to $N-1$}
			\STATE $\mathbf{p} \gets \mathbf{h}+  b_j \exp(i 2\pi j \frac{\mathbf{n}}{LN})$
			\STATE $\mathbf{m} \gets \mathbf{h}-  b_j \exp(i 2\pi j \frac{\mathbf{n}}{LN})$
			\STATE $x_j^\ast \gets -\mathrm{sign}(\mathrm{sum}(|\mathbf{p}|^6\! + 18|\mathbf{p}|^4 +72|\mathbf{p}|^2 -|\mathbf{m}|^6\! - 18|\mathbf{m}|^4\! - 72|\mathbf{m}|^2))$
			\IF{$x_j^\ast=1$}
				\STATE $\mathbf{h} \gets \mathbf{p}$
			\ELSE
				\STATE $\mathbf{h}\gets \mathbf{m}$
			\ENDIF
		\ENDFOR
		\RETURN $\mathbf{x}^\ast$
	\end{algorithmic}
\end{algorithm}

\subsection{Cubic Metric}
\label{sec:CM}

Replacing $f(\mathbf{b}\odot\mathbf{Y}^\pm_j)$ with $\eta_{\scriptscriptstyle N}(\mathbf{b}\!\odot\! \mathbf{Y}^\pm_j)$ in \eqref{eqn:gDef2}, we have \cite{AfrasiabiSPAWC2019}
\begin{equation}
	g_j^\pm(\mathbf{b})=\frac{1}{LN} \sum_{n=0}^{LN-1} \mathbb{E} \left[\left|s(n, \mathbf{b}\odot \mathbf{Y}^\pm_j)\right|^6\right].
	\label{eqn:CMexpCompact}
\end{equation}
The expected values are the third moments of $|s(n,\mathbf{b}\odot \mathbf{Y}^\pm_j)|^2$. Following the approximation mentioned in Remark~\ref{remarkApprox}, they can be obtained from the  third derivative of the moment generating function of the $\chi^2$ random variable $|z(n, \mathbf{b}\odot \mathbf{Y}^\pm_j)|^2$ as defined in~\eqref{eqn:Zdef}. That is,
\begin{equation*}
	\mathbb{E} \left[\left|s(n, \mathbf{b}\!\odot\! \mathbf{Y}^\pm_j)\right|^6\right]= \delta_j^6 \ \frac{d^3 M^\pm_{j,n}(t)}{dt^3}\Big|_{t=0}, \ j=0,\ldots,N-N_e-1.
	\label{eqn:expSinZ}
\end{equation*}
Obtaining the derivative and substituting it in \eqref{eqn:CMexpCompact}, we have \cite{AfrasiabiSPAWC2019}
\begin{align*}
	g_j^\pm(\mathbf{b})\!= \!\frac{\sigma^6}{LN}  \!\sum_{n=0}^{LN-1}\!\! \left[(\lambda^\pm_{j,n})^3 \!+ \!18(\lambda^\pm_{j,n})^2 + 72\lambda^\pm_{j,n}+ \!48\right], \quad j=0,\ldots,N-N_e-1,
\end{align*}
and the decision rule in \eqref{eqn:CEGeneralRule} can be written in closed form as
\begin{align}
	x_j^\ast&= -\mathrm{sign}\left(\sum_{n=0}^{LN-1} \left[(\lambda^+_{j,n})^3 + 18(\lambda^+_{j,n})^2 + 72\lambda^+_{j,n}  -(\lambda^-_{j,n})^3 - 18(\lambda^-_{j,n})^2 - 72\lambda^-_{j,n}\right] \right)
	\label{eqn:finalRule}
\end{align}
for $j=0,\ldots,N-N_e-1$. For the sign variables $j=N-N_e, \ldots, N-1$, consider using sample averages as in~\eqref{eq:estimatorGeneral} with a high $Q$, which was the number of realizations of the random sign variables to calculate the conditional expectations.  \saeed{The following to the last chapter?}Simulations have shown that the CE Method delivers the same performance for several nonzero values of $N_e$ as for $N_e=0$. That is, using accurate sample averages for the final sign variables does not improve the performance. 
 
 The application of the CE Method to the Sign Selection problem essentially leads to  the explicit sign decision criteria derived in this section for PAPR and its substitute SE as well as for the SRCM. For better readability, the pseudocode for SRCM reduction is shown in Algorithm~\ref{alg:ce-cm}, where the expected values required for obtaining $\lambda_{j,n}^\pm$ are constructed by adding the contribution of one subcarrier at each iteration (see lines 5 and 6).

\section{Performance analysis}
\label{sec:analysis}
\saeed{compare this with the upperbounds from other papers? how?}
The CE Method guarantees \eqref{eqn:CEguarantee}, which is rewritten here for convenience: \saeedDissertation{use the updated one for pruned CE}
\begin{equation*}
	f(\mathbf{b}\odot\mathbf{x}^\ast) \leq \mathbb{E} [f(\mathbf{b}\odot\mathbf{X})],
	\label{eq:CEprincipleB}
\end{equation*}
for a given $\mathbf{b}$. \saeedDissertation{It implies that
\begin{equation}
	\mathrm{Pr}(f(\mathbf{b}\odot\mathbf{x}^\ast)>\rho) \leq \mathrm{Pr}(\mathbb{E}_\mathbf{X} [f(\mathbf{c}\odot\mathbf{X})]>\rho).
	\label{eqn:CEguaranteedB}
\end{equation}
That is, the CCDF of the reduced metric is smaller or equal to that of the initial expectation.} In order to characterize $f(\mathbf{b}\odot\mathbf{x}^\ast)$, one approach can be to establish a relation between the distribution of the initial expectation $\mathbb{E}_\mathbf{X} [f(\mathbf{B}\odot\mathbf{X})]$ and that of the uncoded metric values $f(\mathbf{B}), \mathbf{B}\in\mathcal{M}^N$. The analysis will be done for PAPR and SRCM with the help of some useful results from the literature. Concerning the SE metric, a relevant analysis would include a relation between SE reduction and the resulting indirect PAPR reduction, which requires further research.


\subsection{PAPR metric}

\begin{theorem}
	\label{thm:PAPRupper}
	For any $\mathbf{b}\in\mathcal{M}^N$, the reduced PAPR value $\theta_N(\mathbf{b}\odot\mathbf{x}^\ast)$ obtained by the CE Method is bounded  in the limit as
	\begin{align*}
		\lim_{N\to\infty} \frac{\theta_N(\mathbf{b}\odot\mathbf{x}^\ast) - a_N}{b_N} \leq \gamma,
	\end{align*}
	where $b_N=\frac{1}{2}$, $a_N= 2 \log N + \log \log N + \log \frac{\pi}{3}$
	and $\gamma \approx 0.577$ is the Euler constant. \saeedDissertation{include $N_f$. proof is easy, see the notebook}
\end{theorem}
\begin{proof}
	Consider the PAPR of the continuous-time OFDM symbols $u(t,\mathbf{b})$ given in~\eqref{eq:ofdmContSymbolDef} which is defined as
	\begin{align*}
		\xi_N(\mathbf{b})=\max_{t\in[0,T)} |u(t,\mathbf{b})|^2.
	\end{align*}
	Clearly, for any finite oversampling factor $L$, 
	\begin{align*}
		\theta_N(\mathbf{b})\leq \xi_N(\mathbf{b}).
	\end{align*}
	Therefore, It directly follows from~\eqref{eqn:CEguarantee} that
	\begin{align*}
		 \frac{\theta_N(\mathbf{b}\odot\mathbf{x}^\ast) - a_N}{b_N} \leq \mathbb{E}_\mathbf{X} \left[\frac{\xi_N(\mathbf{b}\odot\mathbf{X}) - a_N}{b_N} \right]
	\end{align*}
	for any $N$. Therefore \cite{Rudin1987},
	\begin{align}
		 \lim_{N\to\infty} \frac{\theta_N(\mathbf{b}\odot\mathbf{x}^\ast) - a_N}{b_N} \leq \lim_{N\to\infty} \mathbb{E}_\mathbf{X} \left[\frac{\xi_N(\mathbf{b}\odot\mathbf{X}) - a_N}{b_N} \right].
		 \label{eq:PAPRupperIntermStep}
	\end{align}
	In order to obtain the right hand side limit, recall that the covariance functions of $u(t,\mathbf{b}\odot\mathbf{X})$, as emphasized in Remark~\ref{rmk:sameInTheLimit}, was shown to be identical to that of $u(t,\mathbf{B})$ as $N\to\infty$. In addition, Extreme Value Theory \cite{leadbetter1988} has been employed in \cite{WeiKelly2002} to obtain the asymptotic distribution of $\xi_N(\mathbf{B})$ as
	\begin{align*}
		\lim_{N\to\infty} \mathbb{P} \left( \frac{\xi_N(\mathbf{B}) - b_N}{a_N} \leq w \right) = e^{-e^{-w}}.
	\end{align*}
	That is, the appropriately shifted and scaled variable $\xi_N(\mathbf{B})$ has Gumbel distribution in the limit. Consequently, the results of \cite{WeiKelly2002} hold for the asymptotic distribution of $u(t,\mathbf{b}\odot\mathbf{X})$ as well and 
	\begin{align*}
		\lim_{N\to\infty} \mathbb{P} \left( \frac{\xi_N(\mathbf{b}\odot \mathbf{X}) - b_N}{a_N} \leq w \right) = e^{-e^{-w}}.
	\end{align*}
	Finally, the expected value of a random variable with the Gumbel distribution is
	\begin{align*}
		 \lim_{N\to\infty} \mathbb{E}_\mathbf{X} \left[\frac{\xi_N(\mathbf{b}\odot\mathbf{X}) - a_N}{b_N} \right]=\gamma,
	\end{align*}
	which is the right hand side of~\eqref{eq:PAPRupperIntermStep}. This completes the proof.
\end{proof}

The asymptotic result in Theorem~\ref{thm:PAPRupper} shows an upperbound for $\theta_N(\mathbf{B}\odot\mathbf{x}^\ast)$ when shifted by $a_N$ which grows with $N$. This implies an approximate inequality for finite but large $N$, as stated below.
\begin{remark} 
	For large enough $N$, Theorem~\ref{thm:PAPRupper} implies the upperbound 
	\begin{align}
		\theta_N(\mathbf{b}\odot\mathbf{x}^\ast) \leq \log N + \frac{1}{2} \log \log N + K
		\label{eqn:upperPAPR}
	\end{align}
	for any $\mathbf{b}\in\mathcal{M}^N$, where $K=\frac{1}{2} \log \frac{\pi}{3}+ \gamma \simeq 0.59$.
\end{remark}

Since the upperbound of Theorem~\ref{thm:PAPRupper} holds for every $\mathbf{b}\in\mathcal{M}^N$, it is equivalently an upperbound on the worst-case reduced PAPR value, i.e. $\theta_{N}^{\mathrm{max}}=\max_{\mathbf{b}\in\mathcal{M}^N} \theta_N(\mathbf{b}\odot\mathbf{x}^\ast)$.  Except for relatively small $N$, it is not feasible to observe $\theta_N^{\max}$ in the actual performance investigation by computer simulations or in practice. Instead, it is common to measure the \emph{effective reduced PAPR} $\theta_N^{\mathrm{eff}}$ which is defined according to
\begin{equation}
	\mathbb{P}\{\theta_N(\mathbf{B}\odot\mathbf{x}^\ast)> \theta_N^{\mathrm{eff}}\}=0.001.
	\label{eq:effPAPRdef}
\end{equation} Although it can intuitively be expected that $\theta_{N}^{\mathrm{eff}}$ is much smaller than $\theta_N^{\max}$, a formal relation is not available.


\saeedLater{isn't what gerhard alwyas says that the PAPR for large N deviates arbitrarily small from log N in contradiction with that maximum PAPR is always N? log N and N don't get close, do they?}

\subsection{Cubic Metric}


\saeed{See comments of review 2 of the conf version}
The following theorem was previously presented by the authors in \cite{AfrasiabiSPAWC2019}, which obtains a constant upperbound on the reduced SRCM value. 

\begin{theorem}
	\label{thm:SRCMupperbound}
	The reduced SRCM value $\eta_{\scriptscriptstyle N}(\mathbf{b}\odot\mathbf{x}^\ast)$  obtained by the CE Method is bounded in the limit as
	\begin{align}
		\lim_{N\to\infty}\eta_{\scriptscriptstyle N}(\mathbf{b}\odot\mathbf{x}^\ast) \leq 6
		\label{eq:SRCMupper}
	\end{align}
	for any $\mathbf{b}\in\mathcal{M}^N$. \saeedDissertation{$N_f$}
\end{theorem}

\begin{proof}
	\saeedLater{we don't need this, we can build the result on the fixed c case in the same manner, reyleigh and ...} As stated in~\eqref{eqn:CEguarantee}, the CE Method guarantees that
	\begin{equation}
		\eta_{\scriptscriptstyle N}(\mathbf{b}\odot\mathbf{x}^\ast) \leq \mathbb{E}_\mathbf{X} [\eta_{\scriptscriptstyle N}(\mathbf{b}\odot\mathbf{X})]. 
		\label{eq:CEguaranteeCM}
	\end{equation}
	From the definition of SRCM in~\eqref{eqn:SRCM}, we have $\mathbb{E}[\eta_{\scriptscriptstyle N}(\mathbf{b})]=\frac{1}{LN}\sum_{n=0}^{LN-1} \mathbb{E}[|s(n,\mathbf{b})|^6]$.
In addition, it can be concluded from Remark~\ref{rmk:sameInTheLimit}  that the distribution of the discrete-time signal $s(n,\mathbf{b}\odot\mathbf{X})$ in the limit is the same as that of $s(n,\mathbf{B})$. Therefore,
	\begin{equation}
		\lim_{N\to\infty} \mathbb{E}_\mathbf{X} [\eta_{\scriptscriptstyle N}(\mathbf{b}\odot\mathbf{X})] = \lim_{N\to\infty} \mathbb{E} [\eta_{\scriptscriptstyle N}(\mathbf{B})].
		\label{eqn:Elim}
	\end{equation}
	The distribution of $\eta_{\scriptscriptstyle N}(\mathbf{B})$ is studied in \cite{KimCubicMetric2016}, where it is 
shown that
\begin{equation*}
	\lim_{N\to\infty}\mathbb{E}[\eta_{\scriptscriptstyle N}(\mathbf{B})] = 6.
	\label{eqn:RCMlimit}
\end{equation*}
Considering \eqref{eq:CEguaranteeCM}, \eqref{eqn:Elim} and that an inequality between two sequences is preserved in their limits \cite{Rudin1987}, we have
\begin{align*}
		\lim_{N\to\infty}\eta_{\scriptscriptstyle N}(\mathbf{b}\odot\mathbf{x}^\ast) \leq \lim_{N\to\infty}\mathbb{E}[\eta_{\scriptscriptstyle N}(\mathbf{B})],
	\end{align*}
which completes the proof.
\end{proof}

\saeedDissertation{\begin{corollary}
	Let $\mathrm{RCM}^\ast$ belong to the OFDM signal with reduced SRCM for individual OFDM symbols according to the proposed algorithm. Then $\mathrm{RCM}^\ast \leq \mathrm{RCM}$ as $N\to\infty$. 
\end{corollary}
\begin{proof}
As seen before, we have $\mathrm{RCM}^2=\mathbb{E} [\eta_{\scriptscriptstyle N}(\mathbf{B})]=6$ as $N\to\infty$. The same derivation holds for the SRCM-reduced OFDM signal, i.e. $h_r(n)= \sum_{m=0}^{\infty} s(n-mLN,\mathbf{C}_m \odot \mathbf{x}_m^\ast))$, where $\mathbf{x}^\ast_m$ is the solution obtained for $\mathbf{C}_m$. Therefore, it follows from Theorem~\ref{thm:SRCMupperbound} that 
\begin{align*}
	\lim_{N\to\infty} \mathbb{E}[\eta_{\scriptscriptstyle N}(\mathbf{C} \odot \mathbf{x}^\ast)]\leq \lim_{N\to\infty }\mathbb{E}[\eta_{\scriptscriptstyle N}(\mathbf{B})],
\end{align*}
which completes the proof.
\end{proof}}

Recall  that the reduction of RCM is the actual objective sought in reduction of SRCM and that CM is related to RCM by some constants. Clearly, Theorem~\ref{thm:SRCMupperbound} shows an upperbound on the largest or worst-case reduced SRCM. Being equal to the average of the reduced SRCM values, RCM can be expected to be much smaller than the upperbound unless the distribution of $\eta_{\scriptscriptstyle N}(\mathbf{B}\odot\mathbf{x}^\ast)$ is highly concentrated. Similar to the relation of the effective reduced PAPR and the upperbound, further characterization of RCM reduction is not available.

\section{Simulation Results and Discussion}
\label{sec:perf}

In this section, the performance of the CE method in reducing PAPR and CM is examined via simulation results. The performance here refers to the reduction in the metrics of interest achieved by the suboptimal solution to the Sign Selection problem, including the indirect PAPR reduction gained by applying the proposed method to the SE metric.

\subsection{PAPR reduction}


The Complementary Cumulative Distribution Function (CCDF) of $\theta_N(\mathbf{B})$, i.e. $\mathbb{P}(\theta_N(\mathbf{B})>y) $, with $\mathbf{B}\in\mathcal{M}^N$ is commonly used to represent the \emph{uncoded PAPR}, i.e. the PAPR of an unprocessed signal. Accordingly, the reduction performance is reported by the CCDF of  $\theta_N(\mathbf{B}\odot\mathbf{x}^\ast)$ for $\mathbf{B}\in \mathcal{M}^N$. To report the performance in the text,  the \emph{effective PAPR}  is used which is the PAPR value where CCDF equals 0.001.  

\paragraph{Choice of $Q$} To investigate the reliability of the estimations required in the sign selection rule for PAPR given in \eqref{eq:estimator}, the reduction performance gained by several values of $Q$ for $N=64$  is depicted in Fig.~\ref{fig:64-reduction-Q}.  It was observed that the difference for $Q\geq 100$ was insignificant. Consequently, $Q=100$ has been used in the rest of the  simulations. As a side note, a very low value of $Q=5$ was included in the figure to show the unexpectedly acceptable reduction that it provides. \saeedDissertation{this same sign thing has to be visited again. why would it make a difference?}

\begin{figure}[tb]
	\begin{subfigure}{0.5\textwidth}
		\centering
		\includegraphics[width=\columnwidth]{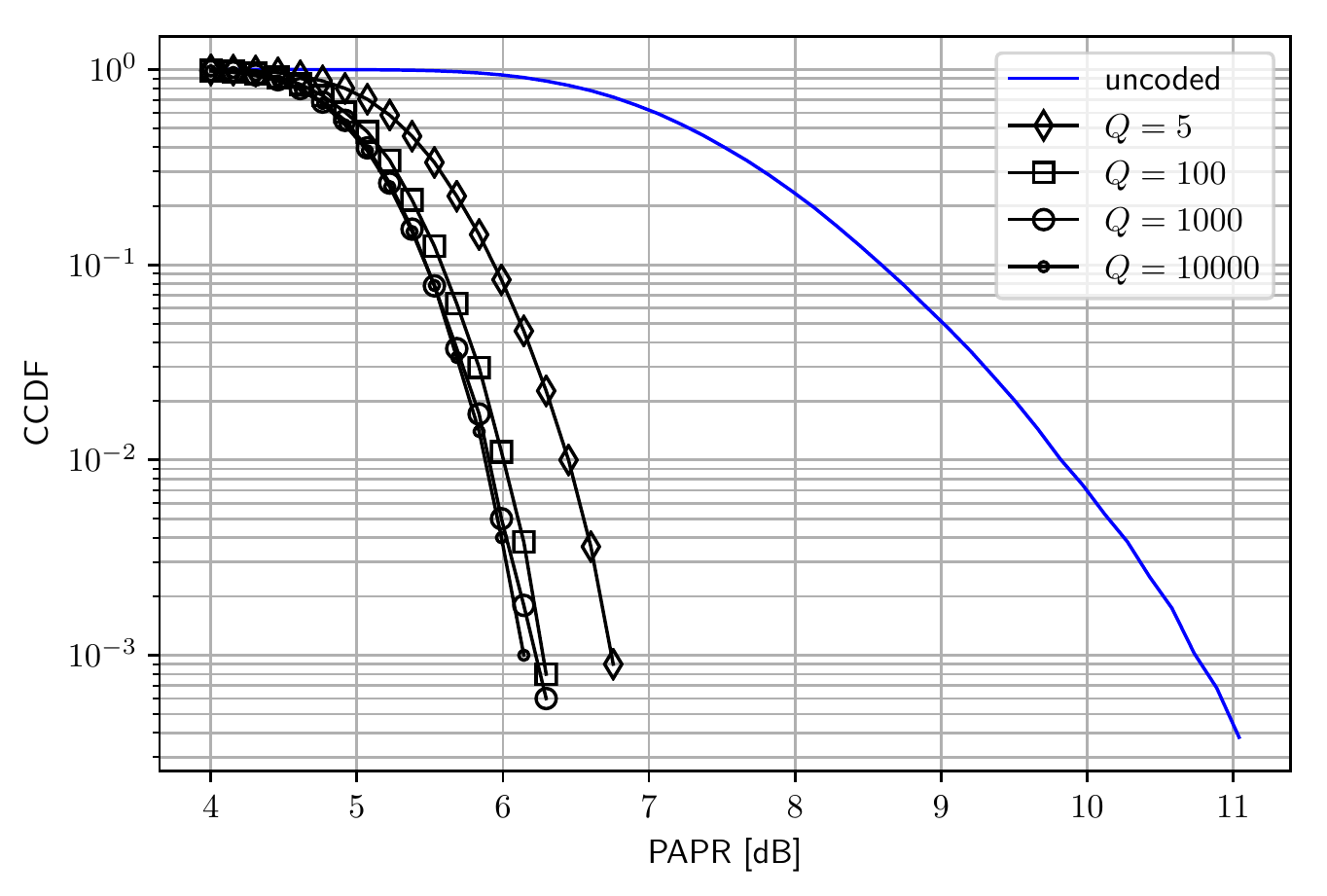}
		\caption{}
		\label{fig:64-reduction-Q}
	\end{subfigure}
	\begin{subfigure}{0.5\textwidth}
		\centering
		\includegraphics[width=\columnwidth]{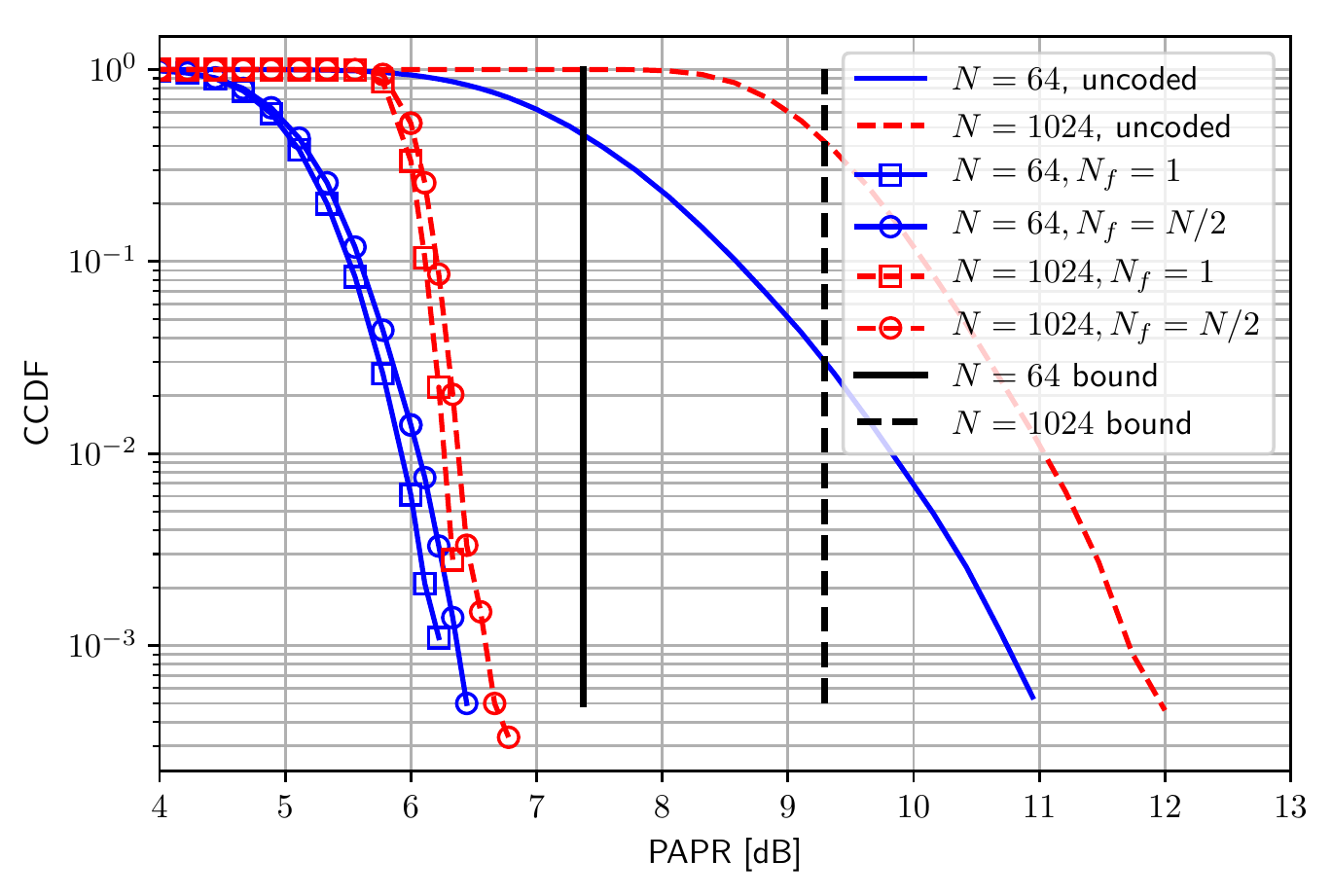}
		\caption{}
		\label{fig:paprCE}
	\end{subfigure}
	\caption{a) Reduction performance of the CE method for PAPR for various number of shots~$Q$, as defined in~\eqref{eq:estimator}, to show the reliability of  estimation for~$N=64$ and~16-QAM.\saeedDissertation{only $Q=5$ is done same sign.} b) Reduction performance of the CE method for PAPR with 16-QAM and~$Q=100$, including analytical upperbound of~\eqref{eqn:upperPAPR} on the worst-case reduced PAPR value.}
\end{figure}

\paragraph{Dependence on $N$} The PAPR reduction performance of the algorithm is shown in Fig.~\ref{fig:paprCE} for $N=64$ and 1024 subcarriers, including its pruned version which will be shortly introduced. The simulation results are depicted only for 16-QAM as similar results were observed for other constellations. \saeed{not all?}  A significant reduction gain of roughly 5.5~dB, equivalently an effective PAPR of 6.5~dB, was observed for $N=1024$.  A noticeable characteristic of the method, evident from the simulations, is that the change in the reduced effective PAPR is relatively small by increasing~$N$ from~$64$ to~1024.

The analytic upper bound on the worst-case reduced PAPR, as shown in Theorem~\ref{thm:PAPRupper} and given in~\eqref{eqn:upperPAPR}, is included in Fig.~\ref{fig:paprCE}. The proof of Theorem~\ref{thm:PAPRupper} relies on the extremal value theory to analyze the expected value of the uncoded PAPR, i.e. $\mathbb{E}[\theta_N(\mathbf{B})]$ as $N\to\infty$. The usefulness of this asymptotic result for $\mathbb{E}[\theta_N(\mathbf{B})]$ with finite $N$ can be asserted as~\eqref{eqn:upperPAPR} is almost equal to a diligently calculated empirical average of $\theta_N(\mathbf{B})$ for $N$ as small as 64. Refer to the discussion in Section~\ref{sec:analysis} regarding the relationship between the upperbound and the effective reduced PAPR. 

\saeedDissertation{
\begin{remark}  
\label{rmk:PAPRanalytic}
A prediction of the amount of reduction in the effective PAPR from the upperbound is not available. Nonetheless, the reduced effective PAPR must be much smaller that the upperbound unless the distribution of the reduced PAPR values is considerably concentrated. This matches the observation made from the curves in Fig.~\ref{fig:paprCE} that the distance of the upperbound from the effective PAPR increases as the accuracy of the asymptotic result, i.e. the value of $N$, increases. The upperbound, which is the expected value of the uncoded OFDM signal, is smaller than the uncoded effective PAPR. Therefore, the existence of such upperbound implies that the reduction gain increases by $N$. \saeed{see if you can turn it into a theorem or sth}
\end{remark}}

\paragraph{Pruned Sign Selection and Rate loss} It has been observed through simulations that the impact of a sign decision increases for the sign variables with higher indices. That is, the reduction steps in the trajectory of the conditional expectations, as the algorithm performs sign decisions for $x_0$ to $x_{N-1}$,  become statistically larger. This motivates pruning the sign bits whose contribution is insignificant. Formally, in the \emph{pruned Sign Selection}, $x^\ast_{0:N_f-1}$ are set as the signs of the first $N_f$ symbols which  fully carry data and the sign bits of $N-N_f$ last symbols are determined by~\eqref{eqn:CEGeneralRule}.

The pruned algorithm with $N_f=\frac{N}{2}$, as shown in Fig.~\ref{fig:paprCE}, causes negligible degradation in the reduction performance while reducing the rate loss of the Sign Selection approach.  Evident from \eqref{eqn:rateLoss}, the rate loss is inversely related to the constellation size $|\mathcal{M}|$. Accordindly, the rate loss is $\frac{1}{12}, \frac{1}{8}$ and $\frac{1}{4}$  for $N_f=\frac{N}{2}$ and 64-QAM, 16-QAM and QPSK respectively. Obviously, a lower rate loss implies a smaller number of sign selections, hence a lower computational complexity.

\paragraph{Indirect PAPR reduction by SE Metric} As discussed in Section~\ref{sec:LSEcalc},  the first $N-N_e$ signs decisions for reduction of the SE metric can be done by the rule in \eqref{eqn:LSErule} and the last $N_e$ are done by \eqref{eq:estimatorGeneral}, where the latter is based on the estimation of the conditional expectations. The choice of $N_e$ depends on $Q$, i.e. the number  of the realizations of the random sign vector used in the estimation. For a given $Q$, a number of the early decisions are done more accurately using closed-form expressions of \eqref{eqn:LSErule}. When the number of remaining signs is low enough, the accuracy of the estimation overcomes. This intuition was evaluated for SE by examining the reduction performance for $N=64$ \saeedDissertation{ , $\kappa=5$} and $Q=10,100, 10000$ for $N_e=0, 5, 10$ and $20$. The relatively small $N$ was chosen on purpose to have a smaller number of total random variables. It was observed that the effective PAPR reduces from roughly 8.5~dB for $N_e=0$ to 6.8~dB for $N_e=10$ which was better than both $5$ and $20$ with effective PAPR of roughly 7.1 and 7 dB. In addition, going from $Q=10$ to $Q=10000$ showed insignificant effect. As a conclusion, $N_e=10$ and $Q=100$ were selected.

\saeedDissertation{The adjustment parameter $\kappa$ introduced in Section~\ref{sec:LSEcalc} was motivated by the observation that having $\kappa$ improved the PAPR reduction for the practical range of $N$. For instance,  the reduced effective PAPR moves from 7.5~dB to 6.5~dB for $N=64$\saeedDissertation{,$k_e=N-10$} and $\kappa=10$. With the same value of $\kappa$, 8.2~dB improved to 6.8~dB for $N=1024$. Therefore, $\kappa=10$ has been used in the simulations.}

The indirect PAPR reduction achieved by reduction of the SE metric is shown in Fig.~\ref{fig:LSEperf} for $N=64$ and 1024. Although increasing the parameter $\kappa$ improves the SE metric in theory, numerical computations limit its value. Thus, $\kappa=10$ was chosen. It can be seen that the indirect PAPR reduction is as strong as the direct one showing a relatively small degradation. The pruning idea works as well, showing that only a slight loss in gain occurs when rate loss is halved.  \saeedLater{would it make a difference it we use same signs for those last signs as in PAPR?}

\begin{figure}[tb]
	\begin{subfigure}{0.5\textwidth}
		\centering
		\includegraphics[width=\columnwidth]{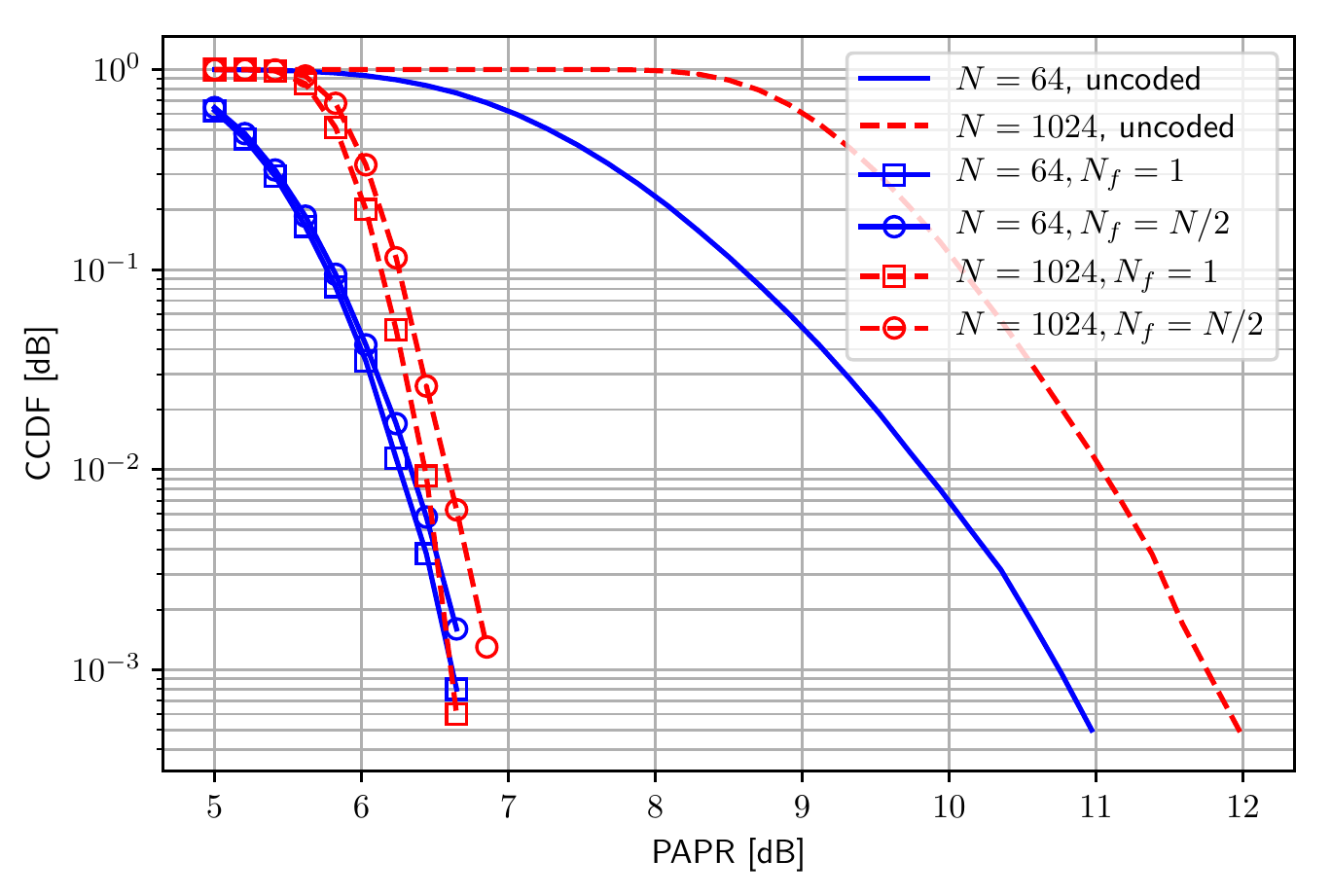}
		\caption{}
		\label{fig:LSEperf}
	\end{subfigure}
	\begin{subfigure}{0.5\textwidth}
		\centering
		\includegraphics[width=\columnwidth]{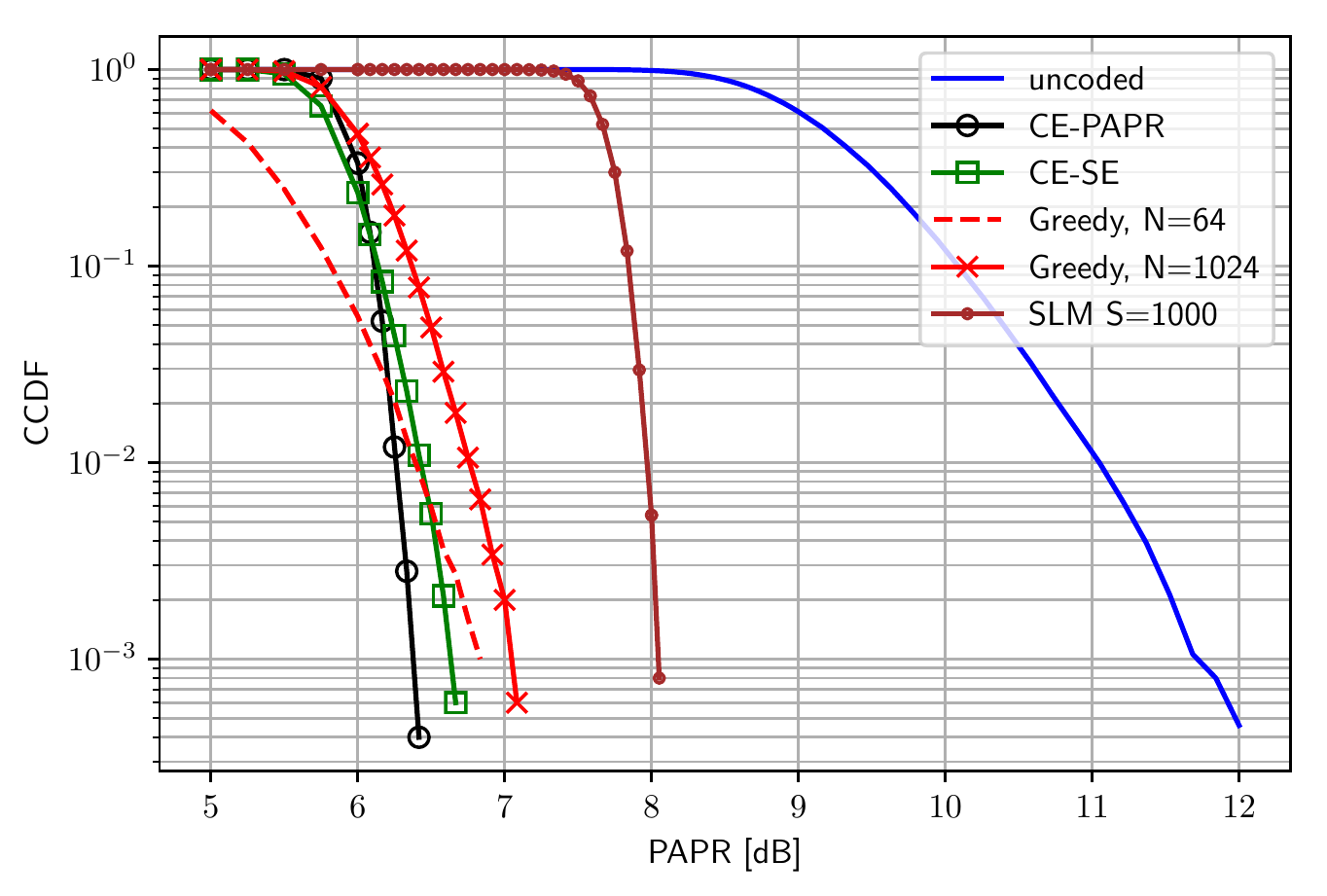}
		\caption{}
		\label{fig:perfComp}
	\end{subfigure}
	\caption{a) Indirect PAPR reduction by pruned application of the CE method to SE for 16-QAM with $\kappa=10, N_e=10 ,Q=100$. b) Reduction performance of the CE method applied to PAPR (CE-PAPR) and SE (CE-SE) compared to Greedy Algorithm~\cite{Sharif2009sign} and Selected Mapping (SLM). $N=1024$ unless stated otherwise.}
\end{figure}


\paragraph{Comparison} It is a rather common characteristic of the PAPR reduction methods in the literature that the reduced PAPR grows larger as $N$ increases. The CE method differs in this regards such that, as mentioned before, the reduced effective PAPR increases only slightly by $N$. 
Among the Sign Selection methods, a competitive proposal referred to as the Greedy Algorithm  \cite{Sharif2009sign} was chosen for comparison. The well-known Selected Mapping (SLM) \cite{543811} with sign flips as phase rotations was also included, which can as well be seen as a Sign Selection method.  The results are gathered in Fig.~\ref{fig:perfComp}, where it can be seen that the Greedy Algorithm performs better for $N=64$ but falls behind for $N=1024$. The performance of SLM depends on the number of independent mappings of the signal denoted by $S$. For the considerably large $S=1000$, the reduction gained by SLM is far lower. As a matter of fact, the performance of SLM can be shown to improve only slightly by increasing $S$ indicating its inherent limitation. The gap becomes larger for higher $N$.


\begin{figure}[tb]
	\begin{subfigure}{0.5\textwidth}
		\centering
		\includegraphics[width=\columnwidth]{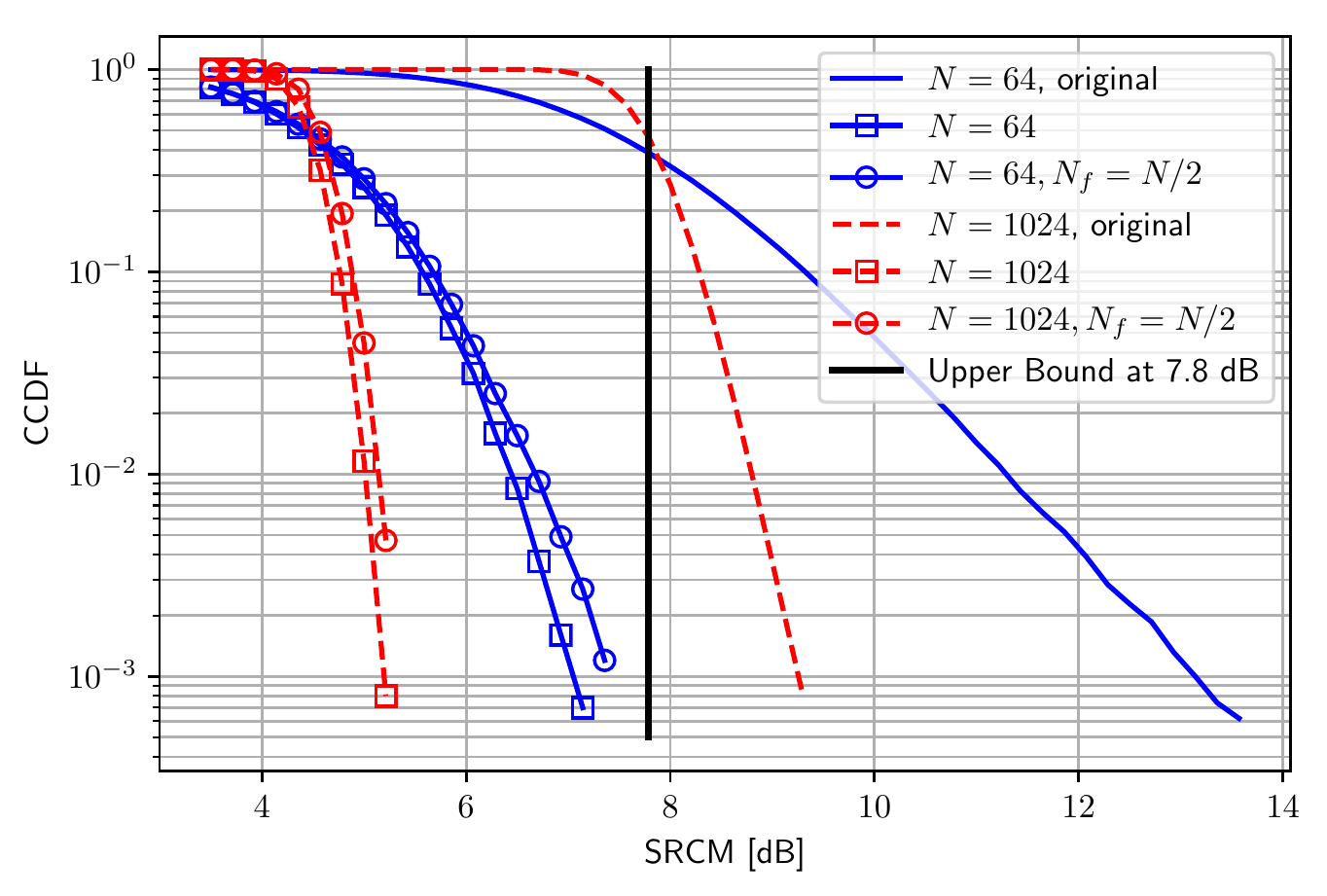}
		\caption{}
		\label{fig:CMreduction}
	\end{subfigure}
	\begin{subfigure}{0.5\textwidth}
		\centering
		\includegraphics[width=\columnwidth]{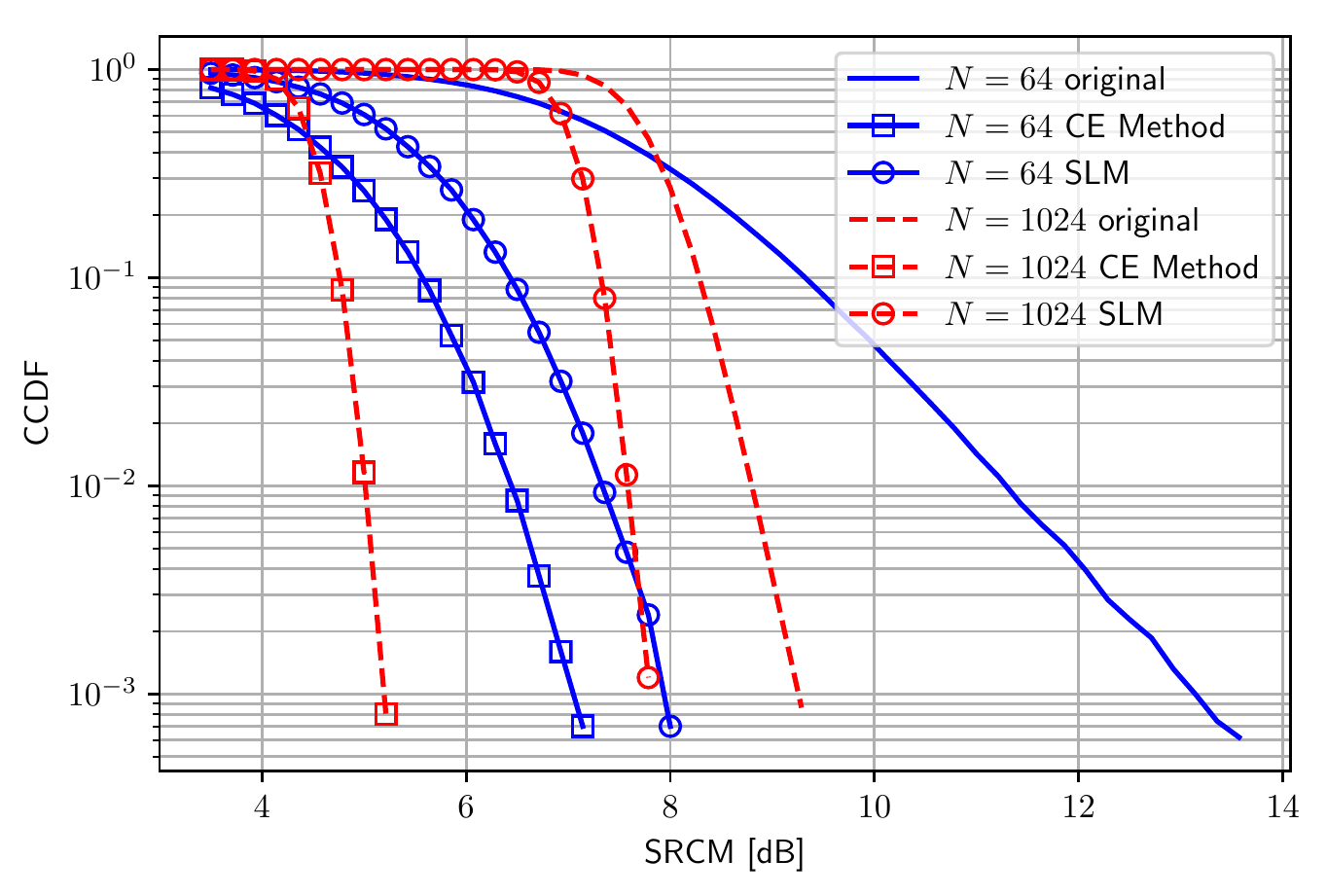}
		\caption{}
		\label{fig:CMreduction2}
	\end{subfigure}
	\caption{a) Reduction performance by pruned CE method applied to SRCM for 16-QAM, including the analytical upperbound on the worst-case reduced metric value. b) Comparison of the SRCM reduction performance of CE Method and SLM with $S=100$.}
\end{figure}


\subsection{Cubic Metric}


Reduction performance for SRCM is shown in Fig.~\ref{fig:CMreduction} for $N=64$ and 1024 to cover a wide range of subcarrier numbers.  As shown in the figure, the performance of the pruned algorithm with $N_f=\frac{N}{2}$, i.e. using the second half of sign bits, is only slightly degraded compared to the $N_f=0$ case. This reconfirms the result seen before in PAPR reduction that  the proposed algorithm provides almost the same reduction by half the full rate loss, i.e. $\frac{1}{2}\log_{|M|} 2$. 

The analytical upperbound of Theorem~\ref{thm:SRCMupperbound} is as well included in Fig.~\ref{fig:CMreduction}. The upperbound on the worst-case reduced SRCM is the expected value of the uncoded SRCM, i.e. $\mathbb{E}[\eta_N(\mathbf{B})]$ as $N\to\infty$. The reliability of this asymptotic result was observed as it matches very closely with the empirical average of $\eta_N(\mathbf{B})$ particularly when $N$ is larger than 64. Similar to the PPAR reduction, the simulation results show a growing reduction gain as $N$ increases. The difference in the SRCM case is that the upperbound is independent of $N$, therefore the reduced effective SRCM decreases, which implies the growing reduction gain.


Recall that the main metric of interest is CM which is calculated from RCM by knowing hardware-related constants. Therefore, we suffice to reporting RCM, which is the expected value of SRCM as  in~\eqref{eqn:RCM-SRCM}.  RCM is reduced roughly from 7.7~dB to 4.5~dB for both $N=64$ and $1024$. That is, a surprising result of nearly 3.2~dB reduction practically regardless of $N$. For $N=512$, which is the case studied in \cite{CMmotorola} with available $K_\mathrm{slp}$ and $K_\mathrm{bw}$, the CM is reduced to 2.87~dB. The available values are presented in Table~\ref{tbl:CMperf}. 


Due to the scarcity of research on CM reduction, we sufficed to the well-known SLM method~\cite{543811} for comparison. 
The result is shown in Fig.~\ref{fig:CMreduction2} for the relatively large $S=100$. For both cases of $N=64$ and $1024$, performance of the proposed algorithm is significantly better than~SLM. 


\begin{table}[bt]
	\centering
	\caption{RCM Reduction Performance.}
	\begin{tabular}{|l|c|c|c|c|c|}
		\hline 
		$N$ & original RCM & original CM & reduced RCM  & reduced CM \\
		\hline
		64 & 7.7 dB& -& 4.5 dB& - \\
		512 & 7.8 dB& 4.8 dB&4.5 dB & 2.87 dB\\
		1024 & 7.8 dB& - & 4.5 dB& -\\
		\hline
	\end{tabular}
	\label{tbl:CMperf}
\end{table}

\section{Conclusion}

%
%

The Method of Conditional Expectations was proposed to find a suboptimal solution to the Sign Selection problem. This investigation led to three particular observations. Firstly, using the conditional expectations as the core element of the sign selection rules provides room for reducing complexity of the algorithm. In particular, proposal of the SE metric as a surrogate function to PAPR led to closed-form expressions for sign selection rule and negligible loss in performance. A similar observation was done for CM which inherently has  a tractable definition. This motivates creativity in developing surrogate functions to replace the metrics with physical significance, i.e. PAPR and CM. Secondly, the structure of the CE Method permits derivation of a meaningful upperbound on the largest reduced metric value, such that it actually guarantees a minimum reduction on the effective metric value. Thirdly, the actual performance observed by simulations show a remarkable reduction which is persistent as $N$ increases. In addition, the reduction gain deteriorates only slightly when reducing the number of used sign bits to half, which implies a significantly lower rate loss.

\appendices

\section{Proof of Theorem~\ref{thm:mcdiarmid}}
\label{app:mcdiarmidProof}

Recall the random vectors $\mathbf{X}^l\in \{-1,1\}^{N-j-1}, l=1,\ldots,Q$ with independent elements as used in the definition of $\hat{g}_j^\pm(\mathbf{b},\mathbf{X}^{1:Q})$ in~\eqref{eq:estimator}. Suppose that the real-valued function $\hat{g}_j^+$ satisfies
\begin{equation}
	\left|\hat{g}_j^+(\mathbf{b},\mathbf{v}^{1:Q})-\hat{g}_j^+(\mathbf{b},\mathbf{z}^{1:Q})\right|\leq d_{m,k}
	\label{eq:McDiarmid-diff-1}
\end{equation}
when vectors $\mathbf{v}^l, \mathbf{z}^l \in\{-1,1\}^{N-j-1}, l=1, 2, \ldots, Q$  disagree only at $v^m_k=-z^m_k$. Then for any $\epsilon \geq 0$, McDiarmid's \emph{independent bounded differences} inequality \cite[p. 206]{mcdiarmid1998} holds as
\begin{equation*}
	\mathbb{P}\left(\left|\hat{g}_j^+(\mathbf{b},\mathbf{X}^{1:Q})- g_j^+(\mathbf{b})\right|\geq \epsilon\right) \leq 2 e^{-2\epsilon^2(\sum_{m,k} d_{m,k}^2)^{-1}},
\end{equation*}
where $g_j^+(\mathbf{b})=\mathbb{E}[\hat{g}_j^+(\mathbf{b},\mathbf{X}^{1:Q})]$.

The bounded differences of \eqref{eq:McDiarmid-diff-1} on $\hat{g}^+_j$ can be shown as follows. 
\begin{align}
	\left|\hat{g}^+_j(\mathbf{b},\mathbf{v}^{1:Q})-\hat{g}^+_j(\mathbf{b},\mathbf{z}^{1:Q})\right| &=\left|\frac{1}{Q} \sum_{l=1}^Q \left[\theta_N\left(\mathbf{b}\odot\psi_{j}^+(\mathbf{v}^l)\right)-\theta_N\left(\mathbf{b}\odot\psi_{j}^+(\mathbf{z}^l)\right) \right] \right| \nonumber \\
	&\leq  \frac{1}{Q} \sum_{l=1}^{Q} \left| \max_n \left|s\left(n,\mathbf{b}\odot\psi_{j}^+(\mathbf{v}^l)\right)\right| - \max_n \left|s\left(n,\mathbf{b}\odot\psi_{j}^+(\mathbf{z}^l)\right)\right|  \right| \nonumber \\
	&\leq \frac{1}{Q} \sum_{l=1}^{Q}  \max_n \left| s\left(n,\mathbf{b}\odot\psi_{j}^+(\mathbf{v}^l)\right) - s\left(n,\mathbf{b}\odot\psi_{j}^+(\mathbf{z}^l)\right)  \right| \nonumber \\
	& = \frac{1}{Q} \frac{1}{ \sqrt{N}\sigma_b}2|b_m|  \leq \frac{1}{Q} \frac{1}{ \sqrt{N}} d
\end{align}
where 
\begin{equation*}
	\left|\max_n |p(n)| - \max_n |q(n)| \right|\leq\max_n |p(n)-q(n)|
\end{equation*}
is used and $d= 2\sigma_b^{-1}\max_{x\in\mathcal{M}} |x|$. Therefore, $\sum_{m,k} d_{m,k}^2 = Q\frac{N-j-1}{N}d$, which completes the proof for $\hat{g}^+_j$. Similar steps can be taken to proof the result for $\hat{g}^-_j$.


\section{Proof of Lemma \ref{lem:variances}}
\label{app:variances}

\saeedDissertation{I didn't take this path: Recall from the Sign Selection problem statement that $\mathbb{E}[\mathbf{B}]\neq  0$ for $\mathbf{B}\in \mathcal{M}^N$ with any choice of $\mathcal{M}'$. In this proof, however, it is necessary to have zero expected value for the data symbols. Consider $\mathbf{G}=\mathbf{B}\odot \mathbf{V}$ for a random $\mathbf{V}\in\{-1,1\}^N$. Clearly, $\mathbb{E}[\mathbf{G}]=0$. Replacing $C$ with $G$ in the derivations throughout the paper does not alter the decision rules and the solution $\mathbf{x}^\ast$ obtained for a realization of $\mathbf{B}$ will remain the same.\saeed{yes? and why if so? and doesn't make the whole thing just trivial?}}

\saeed{from $j+1$?}

\subsection{ $R_{ri}(\mathbf{B},\tau)$ a Gaussian random variable}  We begin the proof by analyzing  $R^j_{ri}(\tau,\mathbf{B})$ at iteration $j$ of the CE Method and for the random vector of data symbols $\mathbf{B}\in\mathcal{M}^N$ which was defined in Lemma~\ref{lem:variances} and is rewritten here  as
\begin{equation}
	R^j_{ri}(\tau,\mathbf{B})= \lim_{N\to\infty} h^{ri}_N(\mathbf{B},t_1,t_2),
	\label{eq:riCorrDef}
\end{equation}
where $\tau=t_2-t_1$ and
\begin{equation}
	h^{ri}_N(\mathbf{B},t_1,t_2)= \mathbb{E}_{\mathbf{Y}^\pm_j} [\hat{u}_r(t_1, \mathbf{B}\!\odot\! \mathbf{Y}^\pm_j)\hat{u}_i(t_2, \mathbf{B}\!\odot\! \mathbf{Y}^\pm_j)].
	\label{eq:hriNdef}
\end{equation}
Based on the definition of the signal $u$ in \eqref{eq:ofdmContSymbolDef}, 
\begin{align}
	h^{ri}_N(\mathbf{B},t_1,t_2)&= \frac{1}{N\sigma_b^2}   \mathbb{E}_{X_{j+1:N-1}} \left[ \sum_{k_1=j+1}^{N-1} \sum_{k_2=j+1}^{N-1}   X_{k_1}\mathrm{Re} \{B_{k_1} e^{i\omega_{k_1} t_1}\}     X_{k_2}\mathrm{Im} \{B_{k_2} e^{i\omega_{k_2} t_2}\} \right] \nonumber \\
	&=\!\frac{1}{N\sigma_b^2} \sum_{k=j+1}^{N-1} \!\mathrm{Re} \{B_k e^{i\omega_k t_1}\} \mathrm{Im} \{B_k e^{i\omega_k t_2}\},
	\label{eqn:rawDefHri}
\end{align}
where $\omega_k=\frac{2\pi}{N}F_s k$ and the independence of the random sign variables $X_{j+1:N-1}$ in $\mathbf{Y}^\pm_j$ is used. 

At this juncture, the relation of $j$ and $N$ must be reviewed. Consider two cases: If $j$ remains constant while $N$ grows, it can be easily seen from the following derivations that the desired quantities are identical in the limit, i.e. as $N\to\infty$, to the case where no sign decision is made by the CE Method. The second case is when $j$ grows with $N$, which needs attention and is the assumption in Lemma~\ref{lem:variances}. Specifically, as introduced in Lemma~\ref{lem:variances}, $j=\rho N$ where $0\leq \rho\leq 1$ is a  constant rational number.

Since the summands in~\eqref{eqn:rawDefHri} are independent, it is straightforward to apply the Central Limit Theorem to show that $h^{ri}_N(\mathbf{B},t_1,t_2)$ converges in distribution to a Gaussian random variable as $N\to\infty$. That is, 
\begin{equation}
	R^j_{ri}(\tau,\mathbf{B})  \sim \mathcal{N}\left(\mu_{ri}(\tau),\sigma^2_{ri}\right),
\end{equation}
where
\begin{equation}
	\mu_{ri}(\tau)= \lim_{N\to\infty} \mathbb{E}\left[h^{ri}_N(\mathbf{B},t_1,t_2)\right] 
	\label{eq:muDef}
\end{equation}
and
\begin{equation}
	\sigma^2_{ri}= \lim_{N\to\infty} \mathbb{E}\left[\left(h^{ri}_N(\mathbf{B},t_1,t_2)\right)^2\right] - \left(\mu_{ri}(\tau)\right)^2.
\end{equation}
Next we derive $\mu_{ri}(\tau)$ and show that $\sigma_{ri}^2=0$, which implies that $R^j_{ri}(\tau,\mathbf{B})$ is equal to $\mu_{ri}(\tau)$  with probability one. 

\subsection{Convergence of $\mu_{ri}(\tau)$} 

Given the independence of the data symbols, we have
\begin{align*}
	\mathbb{E}[h^{ri}_N(\mathbf{B},t_1,t_2)] & = \frac{1}{2N} \sum_{k=j+1}^{N-1} a_k,
\end{align*}
where 
\begin{align*}
	a_k= \cos(\omega_k t_1) \sin(\omega_k t_2) + \sin(\omega_k t_1) \cos(\omega_k t_2).
\end{align*}
Consequently,
\begin{align}
	\mathbb{E}[  h^{ri}_N(\mathbf{B},  t_1,t_2)]  &= \frac{1}{2N} \sum_{k=j+1}^{N-1} \sin(\omega_k \tau) \nonumber \\
	&=\frac{1}{2N} \sum_{k=0}^{N-1}  \sin(\omega_k \tau)  - \frac{1}{2} \frac{j}{N} \frac{1}{j} \sum_{k=0}^{j} \sin(\omega_k \tau).
	\label{eqn:twoSeriesForJ}
\end{align}
where $\tau=t_2-t_1$. Consider the series
\begin{equation}
	 \alpha_M= \frac{1}{2M} \sum_{k=0}^{M-1} \sin\left(\frac{2\pi}{M} k F_s \tau\right) 
\end{equation}
which can be shown to converge as 
\begin{equation*}
	\lim_{M\to\infty} \alpha_M = \frac{1}{4\pi F_s \tau} \left(1-\cos(2\pi F_s \tau)\right), \quad \tau \neq 0.
\end{equation*}
Recall that $j=\rho N$, where $\rho=\frac{m}{p}$ is an irreducible fraction, dictates that $N$ grows as $N=m,2m,\ldots$ with $m\in\mathbb{N}$. Consequently, the first series in \eqref{eqn:twoSeriesForJ}, i.e.
\begin{equation*}
	\beta_N= \frac{1}{2N} \sum_{k=0}^{N-1} \sin(\frac{2\pi}{N} k F_s \tau), \quad N=m,2m,\ldots
\end{equation*}
is a subsequence of $\{\alpha_M\}$, which readily shows that \cite{Rudin1987}
\begin{equation*}
	\lim_{N\to\infty} \beta_N = \lim_{M\to\infty} \alpha_M.
\end{equation*}
Rewriting $\frac{2\pi}{N}$ as $\frac{2\pi}{j} \frac{j}{N}=\frac{2\pi}{j} \rho$, the second series in \eqref{eqn:twoSeriesForJ} can be written as
\begin{equation*}
	\zeta_j=\frac{1}{j} \sum_{k=0}^{j} \sin(\frac{2\pi}{j} \rho F_s k \tau),\quad j=p,2p,\ldots.
\end{equation*}
Then $\zeta_j$ is a subsequence of a sequence similar to $\{\alpha_M\}$ and consequently
\begin{equation*}
	\lim_{j\to\infty} \zeta_j = \frac{1}{4\pi F_s \rho \tau} (1-\cos(2\pi  F_s \rho \tau)), \quad \tau\neq 0
\end{equation*}
Therefore, \eqref{eqn:twoSeriesForJ} converges. Substituting the limit in \eqref{eq:muDef}, we have 
\begin{align}
	&\mu_{ri}(\tau) = 
	\begin{dcases}
		  \frac{1}{4\pi F_s \tau} \left(\frac{1}{\rho}\cos(2\pi F_s \, \rho\,\tau)\!-\! \cos(2\pi F_s \tau)\right) & \tau\neq 0 \\
		0 & \tau=0
	\end{dcases}
\end{align}
where the case of $\tau=0$ is trivial. 

\subsection{Convergence of $\sigma_{ri}^2$ to zero}
Consider that
\begin{align*}
	\left(h^{ri}_N(\mathbf{B},t_1,t_2)\right)^2 &=\frac{1}{N^2 \sigma_b^4} &\Bigg( \sum_{k=j}^{N-1} \Big[B_k^r \cos(\omega_k  t_1) -B_k^i \sin(\omega_k t_1)\Big]   \Big[B_k^r \sin(\omega_k  t_2) -B_k^i \cos(\omega_k t_2)\Big] \Bigg)^2 
\end{align*}
where $B_k^r=\mathrm{Re}[B_k]$ and $B_k^i=\mathrm{Im}[B_k]$. By some manipulations which are omitted for the sake brevity, we have
\begin{align}
	\mathbb{E} \left[\left(h^{ri}_N(\mathbf{B},t_1,t_2)\right)^2\right] =&\ \frac{1}{N^2 \sigma_b^4} \Bigg(\frac{\sigma_b^4}{4}\Big(\sum_{k=j}^{N-1} a_k \Big)^2 -  \frac{\sigma_b^4}{4}\sum_{k=j}^{N-1} a_k^2  \nonumber \\
	& + \gamma \sum_{k=j}^{N-1} \left[ \left(\cos(\omega_k t_1) \sin(\omega_k t_2)\right)^2 + \left(\cos(\omega_k t_2) \sin(\omega_k t_1)\right)^2\right] \nonumber \\ 
	& - \frac{\sigma_b^4}{2} \sum_{k=j}^{N-1} \cos(\omega_k t_1) \cos(\omega_k t_2) \sin(\omega_k t_1) \sin(\omega_k t_2) \Bigg),
	\label{eq:EhriN2}
\end{align}
and $\gamma=\mathbb{E}[(B^r_k)^4]$. Notice that all summands in the four summations of \eqref{eq:EhriN2} are bounded. For instance, $\frac{1}{N^2} \sum_{k=j}^{N-1} a_k^2 \leq \frac{1}{N} A$ 
with $a_k^2\leq A$ for some  $A \geq 0$. Consequently, the non-negative series $\frac{1}{N^2} \sum_{k=j}^{N-1} a_k^2$ converges to zero. By the same argument, the third and fourth summations in \eqref{eq:EhriN2} vanish in the limit too. That is,
\begin{align}
	\lim_{N\to\infty} \mathbb{E}\left[\left(h^{ri}_N(\mathbf{B},t_1,t_2)\right)^2\right] =\lim_{N\to\infty}  \Bigg(\frac{1}{2 N}\sum_{k=j}^{N-1} a_k \Bigg)^2.
\end{align}
It was already shown in derivation of $\mu_{ri}(\tau)$ that $\frac{1}{N}\sum_{k=j}^{N-1}a_k$ converges. Therefore, \cite{Rudin1987}
\begin{align*}
	\lim_{N\to\infty}  \Bigg(\frac{1}{2 N}\sum_{k=j}^{N-1} a_k \Bigg)^2 =    \Bigg( \lim_{N\to\infty} \frac{1}{2 N}\sum_{k=j}^{N-1} a_k \Bigg)^2 	=\Big(\mu_{ri}(\tau)\Big)^2.
\end{align*}
Finally, 
\begin{align}
	\sigma_{ri}^2&=\lim_{N\to\infty}  \Bigg(\frac{1}{2 N}\sum_{k=j}^{N-1} a_k \Bigg)^2 - \Big(\mu_{ri}(\tau)\Big)^2=0.
	\label{eq:sigmaRIzero}
\end{align}

Consequently, we have shown that $R_{ri}(\mathbf{B},\tau)$ is an \emph{almost surely constant} random variable and $R_{ri}(\mathbf{B},\tau)=\mu_{ri}(\tau)$ with probability one. This completes the proof for $R_{ri}(\mathbf{B}, \tau)$.

\subsection{$R_{rr}(\mathbf{B}, \tau)$ and $R_{ii}(\mathbf{B}, \tau)$}

Similarly, for $R_{rr}(\mathbf{B}, \tau)$ we have
\begin{align}
	h^{rr}_N(\mathbf{B},t_1,t_2)&= \mathbb{E} \left[\hat{u}_r(t_1, \mathbf{B}\!\odot\! \mathbf{Y}^\pm_j) \ \hat{u}_r(t_2, \mathbf{B}\!\odot\! \mathbf{Y}^\pm_j)\right] =\frac{1}{N\sigma_b^2} \sum_{k=j}^{N-1} \mathrm{Re} \{b_k e^{i\omega_k t_1}\} \mathrm{Re} \{b_k e^{i\omega_k t_2}\} \nonumber
\end{align}
with
\begin{align}
	\mu_{rr}(\tau)&= \lim_{N\to\infty} \mathbb{E}\left[h^{rr}_N(\mathbf{B},t_1,t_2)\right]  \nonumber \\
	&= \frac{1}{2}\lim_{N\to\infty} \frac{1}{N} \sum_{k=j}^{N-1} \cos(\omega_k \tau) \nonumber \\
	&= \frac{1}{2}\left(\mathrm{sinc} (2F_s\tau)- \rho \mathrm{sinc} (2F_s\tau \rho)\right).
\end{align}
Following the steps taken to derive \eqref{eq:sigmaRIzero}, we have 
\begin{equation*}
	\sigma^2_{rr}= \lim_{N\to\infty} \mathbb{E}\left[\left(h^{rr}_N(\mathbf{B},t_1,t_2)\right)^2\right] - \left(\mu_{rr}(\tau)\right)^2=0,
\end{equation*}
which implies that $R_{rr}(\mathbf{B},\tau)=\mu_{rr}(\tau)$ with probability one and completes the proof. Finally, the derivations for $R_{ii}(\mathbf{B}, \tau)$ are identical to that of $R_{rr}(\mathbf{B}, \tau)$.

\saeedDissertation{ For $R_{ii}(\tau)$ we have
\begin{align*}
	h^{ii}_N(\mathbf{b},t_1,t_2)&= \mathbb{E} [u_i(t_1)u_i(t_2)] \\
	&=\frac{1}{N} \sum_{k=0}^{N-1} \mathrm{Im} \{c_k e^{i\frac{2\pi}{N}F_s k t_1}\} \mathrm{Im} \{c_k e^{i\frac{2\pi}{N}F_s k t_2}\} \nonumber
\end{align*}
with
\begin{align*}
	\mathbb{E}[h_{ii}(\mathbf{B},t_1,t_2)]&= \lim_{N\to\infty} \mathbb{E}[h^{ii}_N(\mathbf{B},t_1,t_2)] \\
	&= \frac{\sigma_b^2}{2}\lim_{N\to\infty} \frac{1}{N} \sum_{k=j}^{N-1} \cos(\frac{2\pi}{N} F_s k \tau) 
\end{align*}
which is the same as derivations for $R_{rr}(\tau)$ .}

\saeedLater{Question: Would it not be simpler to go with $\mathbb{E}_\mathbf{C\odot\mathbf{X}}[]= \mathbb{E}_\mathbf{B} \mathbb{E}_\mathbf{X}[]$ from scratch?}

\saeedLater{Triangular array}

\bibliographystyle{IEEEtran}
\bibliography{biblio}

\begin{thebibliography}{10}
\providecommand{\url}[1]{#1}
\csname url@samestyle\endcsname
\providecommand{\newblock}{\relax}
\providecommand{\bibinfo}[2]{#2}
\providecommand{\BIBentrySTDinterwordspacing}{\spaceskip=0pt\relax}
\providecommand{\BIBentryALTinterwordstretchfactor}{4}
\providecommand{\BIBentryALTinterwordspacing}{\spaceskip=\fontdimen2\font plus
\BIBentryALTinterwordstretchfactor\fontdimen3\font minus
  \fontdimen4\font\relax}
\providecommand{\BIBforeignlanguage}[2]{{%
\expandafter\ifx\csname l@#1\endcsname\relax
\typeout{** WARNING: IEEEtran.bst: No hyphenation pattern has been}%
\typeout{** loaded for the language `#1'. Using the pattern for}%
\typeout{** the default language instead.}%
\else
\language=\csname l@#1\endcsname
\fi
#2}}
\providecommand{\BIBdecl}{\relax}
\BIBdecl

\bibitem{Pun2007}
M.-o. Pun, M.~Morelli, and C.~C.~J. Kuo, \emph{Multi-Carrier Techniques For
  Broadband Wireless Communications: A Signal Processing Perspectives}.\hskip
  1em plus 0.5em minus 0.4em\relax London, UK, UK: Imperial College Press,
  2007.

\bibitem{Ekstrom2006}
H.~Ekstrom, A.~Furuskar, J.~Karlsson, M.~Meyer, S.~Parkvall, J.~Torsner, and
  M.~Wahlqvist, ``{Technical solutions for the 3G long-term evolution},''
  \emph{IEEE Communications Magazine}, vol.~44, no.~3, pp. 38--45, March 2006.

\bibitem{CMmotorola2004}
Motorola, ``{Comparison of PAR and Cubic Metric for Power De-rating},'' 3GPP
  TSG-RAN WG1 LTE, Tech. Rep., May 2004, tdoc R1-040642.

\bibitem{armstrong2002}
J.~Armstrong, ``{Peak-to-average power reduction for OFDM by repeated clipping
  and frequency domain filtering},'' \emph{Electronics Letters}, vol.~38,
  no.~5, pp. 246--247, Feb 2002.

\bibitem{543811}
R.~Bauml, R.~F.~H. Fischer, and J.~Huber, ``Reducing the peak-to-average power
  ratio of multicarrier modulation by selected mapping,'' \emph{Electronics
  Letters}, vol.~32, no.~22, pp. 2056--2057, Oct 1996.

\bibitem{Tellado:2000}
J.~Tellado, \emph{Multicarrier Modulation with Low {PAR}: Applications to {DSL}
  and Wireless}.\hskip 1em plus 0.5em minus 0.4em\relax Kluwer Academic
  Publishers, 2000.

\bibitem{1421929}
S.~H. Han and J.~H. Lee, ``An overview of peak-to-average power ratio reduction
  techniques for multicarrier transmission,'' \emph{Wireless Communications,
  IEEE}, vol.~12, no.~2, pp. 56--65, April 2005.

\bibitem{gerhard2013SPM}
G.~Wunder, R.~F.~H. Fischer, H.~Boche, S.~Litsyn, and J.-S. No, ``The {PAPR}
  problem in {OFDM} transmission: New directions for a long-lasting problem,''
  \emph{The IEEE signal processing magazine}, vol. abs/1212.2865, 2013.

\bibitem{Behravan2011}
M.~Deumal, A.~Behravan, and J.~L. Pijoan, ``{On Cubic Metric Reduction in OFDM
  Systems by Tone Reservation},'' \emph{IEEE Transactions on Communications},
  vol.~59, no.~6, pp. 1612--1620, June 2011.

\bibitem{ZhuDescClip2013}
X.~Zhu, H.~Hu, and Y.~Tang, ``{Descendent clipping and filtering for cubic
  metric reduction in OFDM systems},'' \emph{Electronics Letters}, vol.~49,
  no.~9, pp. 599 --600, April 2013.

\bibitem{SiohanCM2006}
A.~Skrzypczak, P.~Siohan, and J.~Javaudin, ``{Power Spectral Density and Cubic
  Metric for the OFDM/OQAM Modulation},'' in \emph{2006 IEEE International
  Symposium on Signal Processing and Information Technology}, Aug 2006, pp.
  846--850.

\bibitem{Sharif2004constantPMEPR}
M.~Sharif and B.~Hassibi, ``{Existence of codes with constant PMEPR and related
  design},'' \emph{Signal Processing, IEEE Transactions on}, vol.~52, no.~10,
  pp. 2836--2846, Oct 2004.

\bibitem{Afrasiabi2015derandomized}
S.~Afrasiabi~Gorgani and G.~Wunder, ``{Derandomized multi-block sign selection
  for PMEPR reduction of FBMC waveform},'' in \emph{Vehicular Technology
  Conference (VTC Spring), 2015 IEEE 81th}, May 2015.

\bibitem{TellamburaGuidedSS2018}
L.~Wang and C.~Tellambura, ``{Clipping-Noise Guided Sign-Selection for PAR
  Reduction in OFDM Systems},'' \emph{IEEE Transactions on Signal Processing},
  vol.~56, no.~11, pp. 5644--5653, Nov 2008.

\bibitem{Sharif2009sign}
M.~Sharif, V.~Tarokh, and B.~Hassibi, ``Peak power reduction of {OFDM} signals
  with sign adjustment,'' \emph{Communications, IEEE Transactions on}, vol.~57,
  no.~7, pp. 2160--2166, July 2009.

\bibitem{tellamburaCrossEntropy2008}
L.~Wang and C.~Tellambura, ``{Cross-Entropy-Based Sign-Selection Algorithms for
  Peak-to-Average Power Ratio Reduction of OFDM Systems},'' \emph{IEEE
  Transactions on Signal Processing}, vol.~56, no.~10, pp. 4990--4994, Oct
  2008.

\bibitem{MitzenmacherUpfal2005}
\BIBentryALTinterwordspacing
M.~Mitzenmacher and E.~Upfal, \emph{Probability and computing : randomized
  algorithms and probabilistic analysis}.\hskip 1em plus 0.5em minus
  0.4em\relax New York: Cambridge University Press. [Online]. Available:
  \url{http://opac.inria.fr/record=b1117540}
\BIBentrySTDinterwordspacing

\bibitem{WunderPeakValue2003}
G.~Wunder and H.~Boche, ``Peak value estimation of bandlimited signals from
  their samples, noise enhancement, and a local characterization in the
  neighborhood of an extremum,'' \emph{Signal Processing, IEEE Transactions
  on}, vol.~51, no.~3, pp. 771--780, March 2003.

\bibitem{KimCubicMetric2016}
K.~H. Kim, J.~S. No, and D.~J. Shin, ``{On the Properties of Cubic Metric for
  OFDM Signals},'' \emph{IEEE Signal Processing Letters}, vol.~23, no.~1, pp.
  80--83, Jan 2016.

\bibitem{Boyd}
S.~Boyd and L.~Vandenberghe, \emph{Convex Optimization}.\hskip 1em plus 0.5em
  minus 0.4em\relax New York, NY, USA: Cambridge University Press, 2004.

\bibitem{CMmotorola}
Motorola, ``{Cubic Metric in 3GPP LTE},'' 3GPP TSG-RAN WG1 LTE, Tech. Rep., Jan
  2006, tdoc R1-060023.

\bibitem{Benedetto1999}
S.~Benedetto and E.~Biglieri, \emph{Principles of Digital Transmission: With
  Wireless Applications}.\hskip 1em plus 0.5em minus 0.4em\relax Norwell, MA,
  USA: Kluwer Academic Publishers, 1999.

\bibitem{afrasiabiWSA2016}
S.~Afrasiabi-Gorgani and G.~Wunder, ``{A Versatile PAPR Reduction Algorithm for
  5G Waveforms with Guaranteed Performance},'' in \emph{WSA 2016; 20th
  International ITG Workshop on Smart Antennas; Proceedings of}, March 2016.

\bibitem{mcdiarmid1998}
M.~Habib, C.~McDiarmid, J.~Ramirez-Alfonsin, and B.~Reed, \emph{Probabilistic
  Methods for Algorithmic Discrete Mathematics}.\hskip 1em plus 0.5em minus
  0.4em\relax Springer-Verlag Berlin Heidelberg, 1998.

\bibitem{billingsley1999}
P.~Billingsley, \emph{Convergence of probability measures}, 2nd~ed.\hskip 1em
  plus 0.5em minus 0.4em\relax John Wiley \& Sons Inc., 1999.

\bibitem{AfrasiabiSPAWC2019}
S.~{Afrasiabi-Gorgani} and G.~{Wunder}, ``{The Method of Conditional
  Expectations for Cubic Metric Reduction in OFDM},'' in \emph{2019 IEEE 20th
  International Workshop on Signal Processing Advances in Wireless
  Communications (SPAWC)}, July 2019, pp. 1--5.

\bibitem{Rudin1987}
W.~Rudin, \emph{Real and Complex Analysis, 3rd Ed.}\hskip 1em plus 0.5em minus
  0.4em\relax New York, NY, USA: McGraw-Hill, Inc., 1987.

\bibitem{leadbetter1988}
\BIBentryALTinterwordspacing
M.~R. Leadbetter and H.~Rootzen, ``Extremal theory for stochastic processes,''
  \emph{Ann. Probab.}, vol.~16, no.~2, pp. 431--478, 04 1988. [Online].
  Available: \url{https://doi.org/10.1214/aop/1176991767}
\BIBentrySTDinterwordspacing

\bibitem{WeiKelly2002}
S.~Wei, D.~L. Goeckel, and P.~E. Kelly, ``{A modern extreme value theory
  approach to calculating the distribution of the peak-to-average power ratio
  in OFDM systems},'' in \emph{2002 IEEE International Conference on
  Communications. Conference Proceedings. ICC 2002 (Cat. No.02CH37333)},
  vol.~3, 2002, pp. 1686--1690 vol.3.

\end{thebibliography}

\end{document}